\documentclass[a4paper,english,fleqn]{llncs}
\makeatletter\let\endnote\@undefined\makeatother
\usepackage[utf8]{inputenc}
\usepackage[T1]{fontenc}
\usepackage{babel}
\useshorthands{"}
\defineshorthand{"=}{\babelhyphen{hard}}
\defineshorthand{"~}{\babelhyphen{nobreak}}
\usepackage{times}
\usepackage[scaled=.8]{helvet}
\usepackage{amsmath,amssymb,stmaryrd,mathrsfs,mathtools}
\usepackage{microtype}
\usepackage{xspace,calc,xifthen}
\usepackage[hidelinks]{hyperref}
\usepackage{csquotes}
\usepackage[capitalise,nameinlink]{cleveref}
\crefname{section}{Sect.}{Sects.}\Crefname{section}{Section}{Sections}
\crefname{figure}{Fig.}{Figs.}\Crefname{figure}{Figure}{Figures}
\crefname{table}{Tab.}{Tabs.}\Crefname{table}{Table}{Tables}
\crefname{appendix}{App.}{Apps.}\Crefname{appendix}{Appendix}{Appendices}
\crefformat{equation}{(#2#1#3)}\crefrangeformat{equation}{(#3#1#4--#5#2#6)}
\crefname{enumi}{}{}\Crefname{enumi}{}{}
\crefformat{enumi}{(#2#1#3)}\crefrangeformat{enumi}{(#3#1#4--#5#2#6)}
\crefname{definition}{Def.}{Defs.}\Crefname{definition}{Definition}{Definitions}
\crefname{lemma}{Lem.}{Lems.}\Crefname{lemma}{Lemma}{Lemmata}
\crefname{proposition}{Prop.}{Props.}\Crefname{proposition}{Proposition}{Propositions}
\crefname{theorem}{Thm.}{Thms.}\Crefname{theorem}{Theorem}{Theorems}
\crefname{corollary}{Cor.}{Cors.}\Crefname{corollary}{Corollary}{Corollaries}
\crefname{example}{Ex.}{Exs.}\Crefname{example}{Example}{Examples}
\crefname{algorithm}{Alg.}{Algs.}\Crefname{algorithm}{Algorithm}{Algorithms}
\usepackage{enumitem}
\usepackage[font={footnotesize}]{subfig}
\usepackage{graphics,graphicx}
\usepackage{tikz}
\usetikzlibrary{automata,decorations.pathmorphing,matrix,shapes,arrows,calc}
\usepackage{xstring,xifthen}
\usepackage{algorithm}
\usepackage[noend]{algpseudocode}
\algnewcommand\algorithmicrequirx{\phantom{\textbf{Require:}}}
\algnewcommand\Requirx{\item[\algorithmicrequirx]}
\algnewcommand\algorithmicbreak{\textbf{break}}
\algnewcommand\algorithmicchoose{\textbf{choose}}
\algnewcommand\Break{\algorithmicbreak{}}
\algnewcommand\Choose{\algorithmicchoose{}}
\algrenewcommand\alglinenumber[1]{\tiny#1}
\def\Statex{\item[]\vspace*{-2.4ex}}

\usepackage{pdfcomment}
\newif\ifshowednotes\showednotestrue
\newcommand*{\ednoteauthor}{EdNote}
\newcommand*{\ednotecomment}{No comment.}
\newcommand*{\myenotezwritemark}[1]{\leavevmode\marginpar{\pdftooltip{\footnotesize\ednoteauthor(#1)}{\ednotecomment}}}
\usepackage[backref=true]{enotez}
\setenotez{list-name={EdNotes}}
\setenotez{mark-cs={\myenotezwritemark}}
\newcommand{\ednote}[2][Ednote]{%
  \ifshowednotes%
    \renewcommand*{\ednoteauthor}{#1}%
    \renewcommand*{\ednotecomment}{#2}%
    \bgroup%
      \fontsize{6pt}{6pt}\selectfont%
      \endnote{\ifthenelse{\equal{#1}{Ednote}}{#2}{#1: #2}}%
    \egroup%
  \fi%
}
\ifshowednotes%
  \AtEndDocument{\printendnotes}%
\fi
\usepackage[%
  final, % for final version
]{showkeys}

\newcommand*{\restrict}[2]{#1\mathnormal{\upharpoonright}#2}
\newcommand*{\limp}{\mathbin{\rightarrow}}

\newcommand*{\trueval}{\mathit{tt}}
\newcommand*{\falseval}{\mathit{ff}}
\newcommand*{\truefrm}{\mathrm{true}}
\newcommand*{\falsefrm}{\mathrm{false}}
\newcommand*{\refinesto}{\mathrel{\rightsquigarrow}}
\newcommand*{\id}[1]{\mathrm{id}_{#1}}

\newcommand*{\hybind}{\mathnormal{\downarrow}}
\newcommand*{\hyat}{\mathnormal{@}}

\newcommand*{\dldia}[1]{\langle#1\rangle}
\newcommand*{\dlbox}[1]{[#1]}

\newcommand*{\nfatslash}{%
  \mathchoice%
    {\mathbin{{\mathnormal{\mkern-8mu\fatslash\mkern0mu}}}}%
    {\mathbin{{\mathnormal{\mkern-8mu\fatslash\mkern0mu}}}}%
    {\mathbin{{\mathnormal{\mkern-5mu\fatslash\mkern2mu}}}}%
    {\mathbin{{\mathnormal{\mkern-4mu\fatslash\mkern1mu}}}}%
}

\newcommand*{\Attr}{A}
\newcommand*{\Evt}{E}
\newcommand*{\Data}{\mathcal{D}}
\newcommand*{\DataSt}{\Omega}
\newcommand*{\datast}{\omega}
\newcommand*{\datamodels}[1]{\models^{\Data}_{#1}}
\newcommand*{\Lab}{\Lambda}
\newcommand*{\lab}[2]{#1\nfatslash#2}
\newcommand*{\stFrm}{\Phi}
\newcommand*{\trFrm}{\Psi}
\newcommand*{\CtrlSt}{C}
\newcommand*{\ctrlst}{c}
\newcommand*{\initpred}{\varphi_0}
\newcommand*{\Rel}{R}
\newcommand*{\Conf}{\Gamma}
\newcommand*{\Trans}{T}
\newcommand*{\Ecomp}{\Theta}
\newcommand*{\bisim}{\sim}

\newcommand*{\Spec}{\mathit{Sp}}
\newcommand*{\Sig}{\Sigma}
\newcommand*{\Ax}{\mathit{Ax}}
\newcommand*{\Mod}{\mathrm{Mod}}
\newcommand*{\reductop}{\mathnormal{|}}
\newcommand*{\reduct}[2]{#1\reductop#2}

\newcommand{\DHL}{\ensuremath{\mathcal{D}^{\downarrow}}\xspace}
\newcommand{\EDHL}{\ensuremath{\mathcal{E}^{\downarrow}}\xspace}

\newcommand*{\EDHLInst}{\EDHL}
\newcommand*{\EDHLSig}{\mathit{Sig}^{\EDHLInst}}
\newcommand*{\EDHLFrm}{\mathrm{Frm}^{\EDHLInst}}
\newcommand*{\EDHLSen}{\mathrm{Sen}^{\EDHLInst}}
\newcommand*{\EDHLStr}{\mathit{Edts}^{\EDHLInst}}
\newcommand*{\EDHLmodels}[1]{\models^{\EDHLInst}_{#1}}

\useshorthands{"}
\defineshorthand{"=}{-\allowhyphens}
\allowdisplaybreaks
\begin{document}

\title{A Hybrid Dynamic Logic for\\ Event/Data-based Systems}
\author{%
  Rolf Hennicker\inst{1}
  \and
  Alexandre Madeira\inst{2}\thanks{%\fontsize{8.5pt}{9.5pt}\selectfont
    Supported by ERDF through COMPETE 2020 and by National Funds through
    FCT with POCI-01-0145-FEDER-016692 and UID/MAT/04106/2019, in a
    contract foreseen in no.s 4--6 of art.~23 of the DL 57/2016, changed
    by DL 57/2017.}
% This author is financed by the European Regional Development Fund
%   (ERDF) through the Operational Programme for Competitiveness and
%   Internationalisation COMPETE 2020 and by National Funds through the
%   Portuguese funding agency FCT within projects
%   POCI-01-0145-FEDER-016692 and UID/MAT/04106/2019 as well as in the
%   scope of the framework contract foreseen in no.s 4--6 of art.~23 of
%   the Decree-Law 57/2016 (Aug.~29), changed by Portuguese Law 57/2017
%   (July~19).}
    \and
  Alexander Knapp\inst{3}
}
\institute{
  Ludwig-Maximilians-Universität München, Germany\\
  \email{hennicke@pst.ifi.lmu.de}\\[.5ex]
\and
    CIDMA, U.\,Aveiro, Portugal \& QuantaLab, U.\,Minho\\
  \email{madeira@ua.pt}\\[.5ex]
\and 
  Universität Augsburg, Germany\\
  \email{knapp@informatik.uni-augsburg.de}
}

\maketitle

\begin{abstract}
We propose $\EDHL$"=logic as a formal foundation for the specification
and development of event-based systems with local data states.  The
logic is intended to cover a broad range of abstraction levels from
abstract requirements specifications up to constructive specifications.
Our logic uses diamond and box modalities over structured actions
adopted from dynamic logic.  Atomic actions are pairs $\lab{e}{\psi}$
where $e$ is an event and $\psi$ a state transition predicate capturing
the allowed reactions to the event.  To write concrete specifications of
recursive process structures we integrate (control) state variables and
binders of hybrid logic.  The semantic interpretation relies on
event/data transition systems; specification refinement is defined by
model class inclusion.  For the presentation of constructive
specifications we propose operational event/data specifications allowing
for familiar, diagrammatic representations by state transition graphs.
We show that $\EDHL$"=logic is powerful enough to characterise the
semantics of an operational specification by a single $\EDHL$"=sentence.
Thus the whole development process can rely on $\EDHL$"=logic and its
semantics as a common basis.  This includes also a variety of
implementation constructors to support, among others, event refinement
and parallel composition.
\end{abstract}

%!TEX root = edts-hdl.tex

\section{Introduction}

Event-based systems are an important kind of software systems which are
open to the environment to react to certain events. A crucial
characteristics of such systems is that not any event can (or should) be
expected at any time. Hence the control flow of the system is
significant and should be modelled by appropriate means. On the other
hand components administrate data which may change upon the occurrence
of an event.  Thus also the specification of admissible data changes
caused by events plays a major role.

There is quite a lot of literature on modelling and specification of
event-based systems.  Many approaches, often underpinned by graphical
notations, provide formalisms aiming at being constructive enough to
suggest particular designs or implementations, like e.g.,
Event-B~\cite{abrial:2013,farrell-monahan-power:wadt:2016}, symbolic
transition systems~\cite{poizat-royer:jucs:2006},
% modal specifications~\cite{bauer-hennicker-wirsing:tcs:2011,bauer-et-al:scp:2014},
and UML behavioural and protocol state
machines~\cite{uml-2.5,knapp-et-al:fase:2015}.  On the other hand, there
are logical formalisms to express desired properties of event-based
systems. Among them are temporal logics integrating state and
event-based styles~\cite{ter-beek-et-al:fmics:2007}, and various kinds
of modal logics involving data, like first-order dynamic
logic~\cite{harel-kozen-tiuryn:2000} or the modal $\mu$-calculus with
data and time~\cite{groote-mousavi:2014}.  The gap between logics and
constructive specification is usually filled by checking whether
\emph{the} model of a constructive specification satisfies certain
logical formulae.

In this paper we are interested in investigating a logic which is
capable to express properties of event/data-based systems on various
abstraction levels in a common formalism.  For this purpose we follow
ideas of~\cite{madeira-et-al:tcs:2018}, but there data states, effects
of events on them and constructive operational specifications (see
below) were not considered.  The advantage of an expressive logic is
that we can split the transition from system requirements to system
implementation into a series of gradual refinement steps which are more
easy to understand, to verify, and to adjust when certain aspects of the
system are to be changed or when a product line
%variety
of similar products has to be developed.

\enlargethispage{.8618pt}
To that end we propose $\EDHL$-logic, a dynamic logic enriched with
features of hybrid logic.  The dynamic part uses diamond and box
modalities over %regular expressions of actions. 
structured actions. Atomic actions are of
the form $\lab{e}{\psi}$ with $e$ an event and $\psi$ a state transition
predicate specifying the admissible effects of $e$ on the data. Using
sequential composition, union, and iteration we obtain
complex actions that, in connection with the modalities, can be used to
specify required and forbidden behaviour. In particular, if $E$ is a
finite set of events, though data is infinite we are able to capture all
reachable states of the system and to express safety and liveness
properties.  But $\EDHL$"=logic is also powerful enough to specify
concrete, recursive process structures by integrating state variables
and binders from hybrid logic~\cite{braeuner:2010}
with the subtle difference that our state variables are used to denote control states only.
We show that the
dynamic part of the logic is bisimulation invariant while the hybrid
part, due to the ability to bind names to states, is not.

An axiomatic specification $\Spec = (\Sigma, \Ax)$ in $\EDHL$ is given
by an event/data signature $\Sigma = (E, A)$, with a set $E$ of
events and a set $A$ of attributes to model local data states, and a
set of $\EDHL$"=sentences $\Ax$, called axioms, expressing requirements.
For the semantic interpretation we use event/data transition systems
(edts).  Their states are reachable configurations $\gamma = (c,
\omega)$ where $c$ is a control state, recording the current
state of execution, and $\omega$ is a local data state, i.e., a
valuation of the attributes.  Transitions between configurations are
labelled by events.  The semantics of a specification $\Spec$ is
``loose'' in the sense that it consists of \emph{all} edts satisfying
the axioms of the specification. Such structures are called models of
$\Spec$.  Loose semantics allows us to define a simple refinement
notion: $\Spec_1$ refines to $\Spec_2$ if the model class of $\Spec_2$
is included in the model class of $\Spec_1$. We may also say that
$\Spec_2$ is an implementation of $\Spec_1$.

Our refinement process starts typically with axiomatic specifications
whose axioms involve only the dynamic part of the logic. Hybrid features
will successively be added in refinements when specifying more concrete
behaviours, like loops.  Aiming at a concrete design, the use of an
axiomatic specification style may, however, become cumbersome since we
have to state explicitly also all negative cases, what the system should
not do.  For a convenient presentation of constructive specifications we
propose operational event/data specifications, which are a kind of
symbolic transition systems equipped again with a model class semantics
in terms of edts.
%(possibly containing only a single model, up to isomorphism). 
We will show that $\EDHL$"=logic, by use of the hybrid
binder, is powerful enough to characterise the semantics of an
operational specification. 
%The ability to give names to visited
%states, together with the modal features to express transitions, makes
%possible a precise description of the whole dynamics of a process in a
%single sentence. 
Therefore we have not really left $\EDHL$"=logic when refining
axiomatic by operational specifications.
% and also the refinement relation
%defined above can be directly applied for justifying the transition from
%an axiomatic to an operational specification (which could even be
%further refined by operational specifications).
Moreover, since several
constructive notations in the literature, including (essential parts of)
Event-B, symbolic transition systems, and UML protocol state machines,
can be expressed as operational specifications, $\EDHL$"=logic provides a
logical umbrella under which event/data-based systems can be developed.

In order to consider more complex refinements we take up an idea of
Sannella and
Tarlecki~\cite{DBLP:journals/acta/SannellaT88,sannella-tarlecki:2012}
who have proposed the notion of constructor implementation.  This is a
generic notion applicable to specification formalisms based on
signatures and semantic structures for signatures.  As both are
available in the context of $\EDHL$"=logic, we complement our approach
by introducing a couple of constructors, among them event refinement and
parallel composition.  For the latter we provide a useful refinement
criterion relying on a relationship between syntactic and semantic
parallel composition.
%which are appropriate for the formal development of event-based systems. 
The logic and the use of the
implementation constructors will be illustrated by a running
example.
%The remainder of this paper is structured as follows: 

Hereafter, in
\cref{sec:logic}, we introduce syntax and semantics of $\EDHL$"=logic.
In \cref{sec:spec}, we consider axiomatic as well as operational
specifications and demonstrate the expressiveness of $\EDHL$-logic. 
Refinement of both types of specifications using several implementation
constructors is considered in \cref{sec:constructor-impl}.
\Cref{sec:conclusions} provides some concluding remarks.
% and indicates future research directions.
Proofs of theorems and facts can be found
% in~\cite{hennicker-madeira-knapp:fase-arXiv:2019}.
in \cref{app:proofs}.

%%% Local Variables:
%%% mode: LaTeX
%%% mode: TeX-PDF
%%% mode: TeX-source-correlate
%%% TeX-master: "edts-hdl.tex"
%%% End:

%!TEX root = edts-hdl.tex

\section{A Hybrid Dynamic Logic for Event/Data Systems}\label{sec:logic}

We propose the logic $\EDHL$ to specify and reason about
event/data-based systems.  $\EDHL$-logic is an extension of the hybrid
dynamic logic considered in~\cite{madeira-et-al:tcs:2018} by taking into
account changing data.  Therefore, we first summarise our underlying
notions used for the treatment of data.  We then introduce the syntax
and semantics of $\EDHL$ with its hybrid and dynamic logic features
applied to events and data.
% and show that the purely dynamic part of $\EDHL$ is bisimulation
% invariant and enjoys the Hennessy-Milner property.

\subsection{Data States}\label{sec:data-states}

We assume given a universe $\Data$ of \emph{data values}.  A \emph{data
  signature} is given by a set $A$ of \emph{attributes}.  An
$A$-\emph{data state} $\omega$ is a function $\omega : A \to \Data$. We
denote by $\DataSt(A)$ the set of all $A$-data states.  For any data
signature $A$, we assume given a set $\stFrm(A)$ of \emph{state
  predicates} to be interpreted over single $A$-data states, and a set
$\trFrm(A)$ of \emph{transition predicates} to be interpreted over pairs
of pre"= and post"=$A$"=data states.  The concrete syntax of state and
transition predicates is of no particular importance for the following.
For an attribute $a \in A$, a state predicate may be $a > 0$; and a
transition predicate e.g.\ $a' = a + 1$, where $a$ refers to the value
of attribute $a$ in the pre"=data state and $a'$ to its value in the
post"=data state.  Still, both types of predicates are assumed to
contain $\mathrm{true}$ and to be closed under negation (written $\neg$)
and disjunction (written $\lor$); as usual, we will then also use
$\falsefrm$, $\land$, etc.
%and $\limp$. 
Furthermore, we assume for each $A_0 \subseteq A$ a
transition predicate $\id{A_0} \in \trFrm(A)$ expressing that the values
of attributes in $A_0$ are the same in pre- and post-$A$-data states.

We write $\omega \datamodels{A} \varphi$ if $\varphi \in
\stFrm(A)$ is satisfied in data state $\omega$; and $(\omega_1,
\omega_2) \datamodels{A} \psi$ if $\psi \in \trFrm(A)$ is
satisfied in the pre"=data state $\omega_1$ and post"=data state
$\omega_2$.  In particular, $(\omega_1, \omega_2) \datamodels{A}
\id{A_0}$ if, and only if, $\omega_1(a_0) = \omega_2(a_0)$ for all $a_0
\in A_0$.

\subsection{$\EDHL$-Logic}\label{sec:edhl}

\begin{definition}%[Event/data signature]
An \emph{event/data signature} (\emph{ed signature}, for short) $\Sigma
= (E, A)$ consists of a finite set of \emph{events} $E$ and a
data signature $A$.  We write $\Evt(\Sigma)$ for $E$ and
$\Attr(\Sigma)$ for $A$.  We also write $\DataSt(\Sigma)$ for
$\DataSt(\Attr(\Sigma))$, $\stFrm(\Sigma)$ for $\stFrm(\Attr(\Sigma))$,
and $\trFrm(\Sigma)$ for $\trFrm(\Attr(\Sigma))$.  The class of ed
signatures is denoted by $\EDHLSig$.
\end{definition}

Any ed signature $\Sigma$ determines a class of semantic structures, the
\emph{event/data transition systems} which are reachable transition
systems with sets of initial states and events as labels on transitions. The
states are pairs $\gamma = (c, \omega)$, called \emph{configurations},
where $c$ is a \emph{control state} recording the current execution
state and $\omega$ is an $\Attr(\Sigma)$"=data state; we write $\ctrlst(\gamma)$
for $c$ and $\datast(\gamma)$ for $\omega$.

\begin{definition}%[Event/data structure]
A $\Sigma$"=\emph{event/data transition system} ($\Sigma$"=\emph{edts},
for short) $M = (\Gamma, R, \Gamma_0)$ over an ed signature $\Sigma$
consists of a set of \emph{configurations} $\Gamma \subseteq C \times
\DataSt(\Sigma)$ for a set of \emph{control states} $C$; a family of
\emph{transition relations} $R = (R_e \subseteq \Gamma \times\allowbreak
\Gamma)_{e \in \Evt(\Sigma)}$; and a non-empty set of \emph{initial
  configurations} $\Gamma_0 \subseteq \{ c_0 \} \times \Omega_0$ for an
\emph{initial control state} $c_0 \in C$ and a set of \emph{initial data
  states} $\Omega_0 \subseteq \DataSt(\Sigma)$ such that $\Gamma$ is
\emph{reachable} via $R$, i.e., for all $\gamma \in \Gamma$ there are
$\gamma_0 \in \Gamma_0$, $n \geq 0$, $e_1, \ldots, e_n \in E(\Sigma)$,
and $(\gamma_i, \gamma_{i+1}) \in R_{e_{i+1}}$ for all $0 \leq i < n$
with $\gamma_n = \gamma$.  We write $\Conf(M)$ for $\Gamma$, $C(M)$ for
$C$, $\Rel(M)$ for $R$, $c_0(M)$ for $c_0$, $\DataSt_0(M)$ for
$\Omega_0$, and $\Conf_0(M)$ for $\Gamma_0$.  The class of $\Sigma$-edts
is denoted by $\EDHLStr(\Sigma)$.
\end{definition}

%In our logic $\EDHL$ we use
Atomic actions are given by expressions of the form $\lab{e}{\psi}$ with
$e$ an event and $\psi$ a state transition predicate.  The intuition is
that the occurrence of the event $e$ causes a state transition in
accordance with $\psi$, i.e., the pre"= and post"=data states satisfy
$\psi$, and $\psi$ specifies the possible effects of $e$.  Following the
ideas of dynamic logic we also use complex, structured actions formed
over atomic actions by union, sequential composition and iteration.  All
kinds of actions over an ed signature $\Sigma$ are called
$\Sigma$"=\emph{event/data actions} ($\Sigma$"=\emph{ed actions}, for
short). The set $\Lab(\Sigma)$ of $\Sigma$"=ed actions is defined by the
grammar
\begin{equation*}
  \lambda ::= \lab{e}{\psi} \,\mid\, \lambda_1 + \lambda_2 \,\mid\, \lambda_1 ; \lambda_2 \,\mid\, \lambda^*
\end{equation*}
where $e \in \Evt(\Sigma)$ and $\psi \in \trFrm(\Sigma)$.  We use the
following shorthand notations for actions: For a subset $F = \{ e_1,
\ldots, e_k \} \subseteq \Evt(\Sigma)$, we use the notation $F$ to
denote the complex action $\lab{e_1}{\mathrm{true}} + \ldots +
\lab{e_k}{\mathrm{true}}$ and $-F$ to denote the action $E(\Sigma)
\setminus F$.  For the action $E(\Sigma)$ we will write
$\boldsymbol{E}$.  For $e \in E(\Sigma)$, we use the notation $e$ to
denote the action $\lab{e}{\mathrm{true}}$ and $-e$ to denote the action
$\boldsymbol{E} \setminus \{ e \}$.  Hence, if $E(\Sigma) = \{ e_1,
\ldots, e_n\}$ and $e_i \in E(\Sigma)$, the action $-e_i$ stands for
$\lab{e_1}{\mathrm{true}} + \ldots + \lab{e_{i-1}}{\mathrm{true}} +
\lab{e_{i+1}}{\mathrm{true}} + \ldots+\lab{e_n}{\mathrm{true}}$.

The actions $\Lab(\Sigma)$ are \emph{interpreted} over a $\Sigma$"=edts
$M$ as the family of relations $(\Rel(M)_{\lambda} \subseteq \Conf(M)
\times \Conf(M))_{\lambda \in \Lab(\Sigma)}$ defined by
\begin{itemize}[leftmargin=*, topsep=3pt, itemsep=1pt]
  \item $\Rel(M)_{\lab{e}{\psi}} = \{ (\gamma, \gamma') \in \Rel(M)_e \mid
(\datast(\gamma), \datast(\gamma')) \datamodels{A(\Sigma)} \psi \}$,
  \item $\Rel(M)_{\lambda_1 + \lambda_2} = \Rel(M)_{\lambda_1} \cup
\Rel(M)_{\lambda_2}$, i.e., union of relations,
  \item $\Rel(M)_{\lambda_1 ; \lambda_2} = \Rel(M)_{\lambda_1} ;
\Rel(M)_{\lambda_2}$, i.e., sequential composition of relations,
  \item $\Rel(M)_{\lambda^*} = (\Rel(M)_{\lambda})^*$, i.e.,
reflexive-transitive closure of relations.
\end{itemize}
\vskip4pt

To define the event/data formulae of $\EDHL$ we assume given a countably
infinite set $X$ of control state variables which are used in formulae
to denote the control part of a configuration.  They can be bound by the
binder operator $\hybind x$ and accessed by the jump operator $\hyat x$
of hybrid logic.  The dynamic part of our logic is due to the modalities
which can be formed over any ed action over a given ed signature.
$\EDHL$ thus retains from hybrid logic the use of binders, but omits
free nominals.  Thus sentences of the logic become restricted to
express properties of configurations reachable from the initial ones.

\begin{definition}%[Formulae and sentences]
The set $\EDHLFrm(\Sigma)$ of \emph{$\Sigma$"=ed formulae} over an ed
signature $\Sigma$ is given by
\begin{equation*}
  \varrho ::= \varphi \,\mid\, x \,\mid\, \hybind x \,.\, \varrho \,\mid\, \hyat x \,.\, \varrho \,\mid\, \dldia{\lambda} \varrho \,\mid\, \truefrm \,\mid\, \neg \varrho \,\mid\, \varrho_1 \lor \varrho_2
\end{equation*}
where $\varphi \in \stFrm(\Sigma)$, $x \in X$, and $\lambda \in
\Lab(\Sigma)$.  We write $\dlbox{\lambda} \varrho$ for
$\neg\dldia{\lambda} \neg\varrho$ and we use the usual boolean
connectives as well as the constant $\falsefrm$ to denote
$\neg\truefrm$.\emph{\footnote{We use $\truefrm$ and $\falsefrm$ for
    predicates and formulae; their meaning will always be clear from the
    context.  For boolean values we will use instead the notations
    $\trueval$ and $\falseval$.}}  The set $\EDHLSen(\Sigma)$ of
\emph{$\Sigma$"=ed sentences} consists of all $\Sigma$"=ed formulae
without free variables, where the free variables are defined as usual
with $\hybind x$ being the unique operator binding variables.
\end{definition}

Given an ed signature $\Sigma$ and a $\Sigma$"=edts $M$, the
satisfaction of a $\Sigma$"=ed formula $\varrho$ is inductively defined
w.r.t.\ valuations $v : X \to C(M)$, mapping variables to control
states, and configurations $\gamma \in \Conf(M)$:
\begin{itemize}[leftmargin=*, topsep=2pt]
  \item $M, v, \gamma \EDHLmodels{\Sigma} \varphi$ iff $\datast(\gamma)
\datamodels{A(\Sigma)} \varphi$;
  \item $M, v, \gamma \EDHLmodels{\Sigma} x$ iff $\ctrlst(\gamma) = v(x)$;
  \item $M, v, \gamma \EDHLmodels{\Sigma} \hybind x \,.\, \varrho$
iff $M, v\{ x \mapsto \ctrlst(\gamma) \}, \gamma \EDHLmodels{\Sigma} \varrho$;
  \item $M, v, \gamma \EDHLmodels{\Sigma} \hyat x \,.\, \varrho$ iff $M,
v, \gamma' \EDHLmodels{\Sigma} \varrho$ for all $\gamma' \in
\Conf(M)$ with $\ctrlst(\gamma') = v(x)$;
  \item $M, v, \gamma \EDHLmodels{\Sigma} \dldia{\lambda}\varrho$ iff
$M, v, \gamma' \EDHLmodels{\Sigma} \varrho$ for some $\gamma' \in \Conf(M)$ with $(\gamma, \gamma')
\in \Rel(M)_{\lambda}$;
  \item $M, v, \gamma \EDHLmodels{\Sigma} \truefrm$ always holds;
  \item $M, v, \gamma \EDHLmodels{\Sigma} \neg\varrho$ iff $M, v, \gamma
\not\EDHLmodels{\Sigma} \varrho$;
  \item $M, v, \gamma \EDHLmodels{\Sigma} \varrho_1 \lor \varrho_2$ iff
$M, v, \gamma \EDHLmodels{\Sigma} \varrho_1$ or $M, v, \gamma
\EDHLmodels{\Sigma} \varrho_2$.
\end{itemize}
If $\varrho$ is a sentence then the valuation is irrelevant.
$M$ \emph{satisfies} a sentence $\varrho \in \EDHLSen(\Sigma)$,
denoted by $M \EDHLmodels{\Sigma} \varrho$,
if $M, \gamma_0 \EDHLmodels{\Sigma}
\varrho$ for all $\gamma_0 \in \Conf_0(M)$.
 
%%Alex counterexample
% Note that $\hyat x \,.\, \varphi$ is not semantically equivalent to $x
% \limp \varphi$: If $M$ is a $\Sigma$-ed model with initial control
% state $c_0$, a single initial data state $\omega_0$ such that
% $\omega_0 \datamodels{A(\Sigma)} \neg\varphi$, another control state
% $c \neq c_0$, and a data state $\omega \in \DataSt(M)$ such that
% $((c_0, \omega_0), (c, \omega)) \in \Rel_e$, then $M
% \not\EDHLmodels{\Sigma} \hybind x \,.\,
% \dldia{\lab{e}{\truefrm}}(\hyat x \,.\, \varphi)$ but $M
% \EDHLmodels{\Sigma} \hybind x \,.\, \dldia{\lab{e}{\truefrm}}(x \limp
% \varphi)$.\ednote[AK]{Check whether this is a counterexample.}

By borrowing the modalities from dynamic
logic~\cite{harel-kozen-tiuryn:2000,groote-mousavi:2014}, $\EDHL$ is
able to express liveness and safety requirements as illustrated in our
running ATM example below. There we use the fact that we can state
properties over all reachable states by sentences of the form
$\dlbox{\boldsymbol{E}^*}\varphi$.  In particular, deadlock-freedom can
be expressed by
$\dlbox{\boldsymbol{E}^*}\dldia{\boldsymbol{E}}\truefrm$.  The logic
$\EDHL$, however, is also suited to directly express process structures
and, thus, the implementation of abstract requirements.  The binder
operator is essential for this.
%The ability to give names to visited
%states, together with the modal features to express transitions, makes
%possible a precise description of the whole dynamics of a process in a
%single sentence.
For example, we can specify a process which switches a
boolean value, denoted by the attribute $\mathsf{val}$, from $\trueval$
to $\falseval$ and back by the following sentence:
\bgroup\abovedisplayskip6pt\belowdisplayskip6pt
\begin{equation*}
  \hybind x_0 \,.\, \mathsf{val} = \trueval \land \dldia{\lab{\mathsf{switch}}{\mathsf{val}' = \falseval}} \dldia{\lab{\mathsf{switch}}{\mathsf{val}' = \trueval}} x_0
\ \text{.}
\end{equation*}
\egroup

\subsection{Bisimulation and Invariance}\label{sec:invariance}

Bisimulation is a crucial notion in both behavioural systems
specification and in modal logics. On the specification side, it
provides a standard way to identify systems with the same behaviour by
abstracting the internal specifics of the systems; this is also
reflected at the logic side, where bisimulation frequently relates
states that satisfy the same formulae.  We explore some properties of
$\EDHL$ w.r.t.\ bisimilarity.
%  Moreover,
% under some additional conditions, it is also usual that states
% satisfying the same formulae are bisimilar (Hennessy-Milner
% theorem).
% Because of the hybrid machinery of $\EDHL$, used to specify concrete
% aspects of models, modal invariance does not hold in general in our
% logic.
% However, we identify a
% fragment of the logic that enjoys that property. Then, the converse
% characterisation is established. 
Let us first introduce the notion of bisimilarity in the context of
$\EDHL$:

\begin{definition}
Let $M_1, M_2$ be $\Sigma$"=edts.  A relation $B \subseteq \Conf(M_1)
\times \Conf(M_2)$ is a \emph{bisimulation relation} between $M_1$ and
$M_2$ if for all $(\gamma_1, \gamma_2) \in B$ the following conditions
hold:
\vskip2pt%\vskip\topsep\vskip\partopsep
\begin{enumerate}[leftmargin=*, topsep=0pt, partopsep=0pt, label={(atom)}, ref={atom}]
  \item\label{it:bisim-atom} for all $\varphi \in
\stFrm(\Sigma)$, $\omega(\gamma_1) \datamodels{A(\Sigma)} \varphi$ iff
$\omega(\gamma_2) \datamodels{A(\Sigma)} \varphi$;
\end{enumerate}
\begin{enumerate}[leftmargin=*, topsep=0pt, partopsep=0pt, label={(zig)}, ref={zig}]
  \item\label{it:bisim-zig} for all $\lab{e}{\psi} \in
\Lab(\Sigma)$ and for all $\gamma_1' \in \Conf(M_1)$ with $(\gamma_1,
\gamma_1') \in R(M_1)_{\lab{e}{\psi}}$, there is a $\gamma_2' \in
\Conf(M_2)$ such that $(\gamma_2, \gamma_2') \in R(M_2)_{\lab{e}{\psi}}$
and $(\gamma_1', \gamma_2') \in B$;
\end{enumerate}
\begin{enumerate}[leftmargin=*, topsep=0pt, partopsep=0pt, label={(zag)}, ref={zag}]
  \item\label{it:bisim-zag} for all $\lab{e}{\psi} \in
\Lab(\Sigma)$ and for all $\gamma_2' \in \Conf(M_2)$ with $(\gamma_2,
\gamma_2') \in R(M_2)_{\lab{e}{\psi}}$, there is a $\gamma_1' \in
\Conf(M_1)$ such that $(\gamma_1, \gamma_1') \in R(M_1)_{\lab{e}{\psi}}$
and $(\gamma_1', \gamma_2') \in B$.
\end{enumerate}
\vskip2pt%\vskip\topsep
$M_1$ and $M_2$ are \emph{bisimilar}, in symbols $M_1 \bisim
M_2$, if there exists a bisimulation relation $B \subseteq \Conf(M_1)
\times \Conf(M_2)$ between $M_1$ and $M_2$ such that
\begin{enumerate}[leftmargin=*, topsep=2pt, partopsep=0pt, label={(init)}, ref={init}]
  \item\label{it:bisim-init} for any $\gamma_1 \in \Conf_0(M_1)$, there
is a $\gamma_2 \in \Conf_0(M_2)$ such that $(\gamma_1,\gamma_2)\in B$
and for any $\gamma_2 \in \Conf_0(M_2)$, there is a $\gamma_1 \in
\Conf_0(M_1)$ such that $(\gamma_1,\gamma_2)\in B$.
\end{enumerate}
\end{definition}

Now we are able to establish a Hennessy-Milner like correspondence for a
fragment of $\EDHL$.  Let us call \emph{hybrid-free sentences of
  $\EDHL$} the formulae obtained by the grammar
%
%\bgroup\abovedisplayskip4pt\belowdisplayskip4pt
\begin{equation*}
  \varrho ::= \varphi \,\mid \, \dldia{\lambda} \varrho \,\mid\, \truefrm \,\mid\, \neg \varrho \,\mid\, \varrho_1 \lor \varrho_2
\ \text{.}
\end{equation*}
%\egroup

\begin{theorem}\label{thm:bisim-invariance}
Let $M_1, M_2$ be bisimilar $\Sigma$-edts.  Then $M_1
\EDHLmodels{\Sigma} \varrho$ iff $M_2 \EDHLmodels{\Sigma} \varrho$ for
all hybrid-free sentences $\varrho$.
\end{theorem}

The converse of \cref{thm:bisim-invariance} does not hold, in general,
and the usual image-finiteness assumption has to be imposed: A
$\Sigma$-edts $M$ is \emph{image-finite} if, for all $\gamma \in
\Conf(M)$ and all $e \in \Evt(\Sigma)$, the set $\{ \gamma' \mid
(\gamma, \gamma') \in R(M)_e \}$ is finite. Then:

\begin{theorem}\label{bisinvariance2}
Let $M_1, M_2$ be image-finite $\Sigma$-edts and $\gamma_1 \in
\Conf(M_1)$, $\gamma_2 \in \Conf(M_2)$ such that
$M_1,\gamma_1\EDHLmodels{\Sigma} \varrho$ iff $M_2,\gamma_2
\EDHLmodels{\Sigma} \varrho$ for all hybrid-free sentences
$\varrho$. Then there exists a bisimulation $B$ between $M_1$ and $M_2$
such that $(\gamma_1,\gamma_2)\in B$.
\end{theorem}

%%% Local Variables:
%%% mode: LaTeX
%%% mode: TeX-PDF
%%% mode: TeX-source-correlate
%%% TeX-master: "edts-hdl.tex"
%%% End:

%!TEX root = edts-hdl.tex

\section{Specifications of Event/Data Systems}\label{sec:spec}
%\squeezeup
\subsection{Axiomatic Specifications}\label{sec:axiomatic}
%\squeezeup
Sentences of $\EDHL$-logic can be used to specify properties of
event/data systems and thus to write system specifications in an
axiomatic way.

\begin{definition}%[Axiomatic specification]
An \emph{axiomatic ed specification} $\Spec = (\Sig(\Spec), \Ax(\Spec))$
in $\EDHL$ consists of an ed signature $\Sig(\Spec) \in \EDHLSig$ and a
set of \emph{axioms} $\Ax(\Spec) \subseteq \EDHLSen(\Sig(\Spec))$.

The \emph{semantics of $\Spec$} is given by the pair $(\Sig(\Spec),
\Mod(\Spec))$ where $\Mod(\Spec) = \{ M \in \EDHLStr(\Sig(\Spec)) \mid M
\EDHLmodels{\Sig(\Spec)} \Ax(\Spec) \}$.  The edts in $\Mod(\Spec)$ are
called \emph{models} of $\Spec$ and $\Mod(\Spec)$ is the \emph{model
  class} of $\Spec$.
\end{definition}

As a direct consequence of \cref{thm:bisim-invariance} we have:

\begin{corollary}\label{cor:bisim-closed}
The model class of an axiomatic ed specification exclusively expressed by
hybrid-free sentences is closed under bisimulation.
\end{corollary}

This result does not hold for sentences with hybrid features. For
instance, consider the specification $\Spec=\big((\{e\},\{a\}),\{\hybind
x \,.\, \dldia{\lab{e}{a'=a}} x\}\big)$: An edts with a single control
state $c_0$ and a loop transition $R_e= \{ (\gamma_0, \gamma_0) \}$ for
$\ctrlst(\gamma_0) = c_0$ is a model of $\Spec$.  However, this is
obviously not the case for its bisimilar edts with two control states
$c_0$ and $c$ and the relation $R'_e=\{ (\gamma_0, \gamma), (\gamma,
\gamma_0) \}$ with $\ctrlst(\gamma_0) = c_0$, $\ctrlst(\gamma) = c$ and
$\datast(\gamma_0) = \datast(\gamma)$.

\begin{example}\label{ex:spec0}
As a running example we consider an ATM.  We start with an abstract
specification $\Spec_0$ of fundamental requirements for its interaction
behaviour based on the set of events $E_0 = \{ \mathsf{insertCard},
\mathsf{enterPIN}, \mathsf{ejectCard}, \mathsf{cancel} \}$\footnote{For
  shortening the presentation we omit further events like withdrawing
  money, etc.} and on the singleton set of attributes $A_0 = \{
\mathsf{chk} \}$ where $\mathsf{chk}$ is boolean valued and records the
correctness of an entered PIN.  Hence our first ed signature is
$\Sigma_0 = (E_0, A_0)$ and $\Spec_0 = (\Sigma_0, \Ax_0)$ where $\Ax_0$
requires the following properties expressed by corresponding axioms
\crefrange{eq:0.1}{eq:0.3}:
\begin{itemize}[leftmargin=*, topsep=2pt]
  \item ``Whenever a card has been inserted, a correct PIN can
  eventually be entered and also the transaction can eventually be
  cancelled.''
\bgroup\abovedisplayskip6pt\belowdisplayskip6pt 
\begin{equation}\tag{0.1}\label{eq:0.1}
  \dlbox{\boldsymbol{E}^*; \mathsf{insertCard}}(\dldia{\boldsymbol{E}^*;\lab{\mathsf{enterPIN}}{\mathsf{chk}' = \trueval}}\truefrm
\land
\dldia{\boldsymbol{E}^*; \mathsf{cancel}}\truefrm)
\end{equation}
\egroup

  \item ``Whenever either a correct PIN has been entered or the
transaction has been cancelled, the card can eventually be ejected.''
\bgroup\abovedisplayskip6pt\belowdisplayskip6pt 
\begin{equation}\tag{0.2}\label{eq:0.2}
  \dlbox{\boldsymbol{E}^*; (\lab{\mathsf{enterPIN}}{\mathsf{chk}' = \trueval})
  + \mathsf{cancel}}\dldia{\boldsymbol{E}^*;\mathsf{ejectCard}}\truefrm
\end{equation}
\egroup

  \item ``Whenever an incorrect PIN has been entered three times in a
row, the current card is not returned.'' This means that the card is
kept by the ATM which is not modelled by an extra event. It may,
however, still be possible that another card is inserted afterwards.  So
an $\mathsf{ejectCard}$ can only be forbidden as long as no next card is
inserted.
\bgroup\abovedisplayskip6pt\belowdisplayskip6pt 
\begin{equation}\tag{0.3}\label{eq:0.3}
  \dlbox{\boldsymbol{E}^*; (\lab{\mathsf{enterPIN}}{\mathsf{chk}' = \falseval})^3;\, (-\mathsf{insertCard})^*; \mathsf{ejectCard}}\falsefrm
\end{equation}
\egroup
where $\lambda^n$ abbreviates the $n$"=fold sequential composition
$\lambda; \ldots; \lambda$.
\end{itemize}
\end{example}

The semantics of an axiomatic ed specification is loose allowing usually
for many different realisations.  A refinement step is therefore
understood as a restriction of the model class of an abstract
specification.  Following the terminology of Sannella and
Tarlecki~\cite{DBLP:journals/acta/SannellaT88,sannella-tarlecki:2012},
we call a specification refining another one an \emph{implementation}.
Formally, a specification $\Spec'$ is a \emph{simple implementation} of
a specification $\Spec$ over the same signature, in symbols $\Spec
\refinesto \Spec'$, whenever $\Mod(\Spec) \supseteq \Mod(\Spec')$.
Transitivity of the inclusion relation ensures gradual step-by-step
development by a series of refinements.

\begin{example}\label{ex:spec1}
We provide a refinement $\Spec_0 \refinesto \Spec_1$ where $\Spec_1 =
(\Sigma_0,\allowbreak \Ax_1)$ has the same signature as $\Spec_0$ and
$\Ax_1$ are the sentences \crefrange{eq:1.1}{eq:1.4} below; the last two
use binders to specify a loop.  As is easily seen, all models of
$\Spec_1$ must satisfy the axioms of $\Spec_0$.
\begin{itemize}[leftmargin=*, topsep=2pt]
  \item ``At the beginning a card can be inserted with the effect that
$\mathsf{chk}$ is set to $\falseval$; nothing else is possible at the
beginning.''
\bgroup\abovedisplayskip6pt\belowdisplayskip6pt 
\begin{equation}\tag{1.1}\label{eq:1.1}
\begin{split}
&\dldia{\lab{\mathsf{insertCard}}{\mathsf{chk}' = \falseval}}\truefrm \land{}\\[-.5ex]
&\dlbox{\lab{\mathsf{insertCard}}{\neg(\mathsf{chk}' = \falseval)}} \falsefrm \land \dlbox{-\mathsf{insertCard}} \falsefrm
\end{split}
\end{equation}
\egroup

  \item ``Whenever a card has been inserted, a PIN can  be entered (directly afterwards) and
also the transaction can be cancelled; but nothing else.''
\bgroup\abovedisplayskip6pt\belowdisplayskip6pt 
\begin{equation}\tag{1.2}\label{eq:1.2}
\begin{split}
  \dlbox{\boldsymbol{E}^*; \mathsf{insertCard}}(&\dldia{\mathsf{enterPIN}}\truefrm \land \dldia{\mathsf{cancel}}\truefrm \land{}\\[-.5ex]
                                                &\dlbox{-\{\mathsf{enterPIN}, \mathsf{cancel}\}}\falsefrm)
\end{split}
\end{equation}
\egroup

  \item ``Whenever either a correct PIN has been entered or the
transaction has been cancelled, the card can eventually be ejected and
the ATM starts from the beginning.''
\bgroup\abovedisplayskip6pt\belowdisplayskip6pt 
\begin{equation}\tag{1.3}\label{eq:1.3}
  \hybind x_0 \,.\,\dlbox{\boldsymbol{E}^*; (\lab{\mathsf{enterPIN}}{\mathsf{chk}' = \trueval}) + \mathsf{cancel}}\dldia{\boldsymbol{E}^*;\mathsf{ejectCard}}x_0
\end{equation}
\egroup

  \item ``Whenever an incorrect PIN has been entered three times in a
row the ATM starts from the beginning.'' Hence the current card is kept.
\bgroup\abovedisplayskip6pt\belowdisplayskip6pt 
\begin{equation}\tag{1.4}\label{eq:1.4}
  \hybind x_0 \,.\, \dlbox{\boldsymbol{E}^*; (\lab{\mathsf{enterPIN}}{\mathsf{chk}' = \falseval})^3}x_0
\end{equation}
\egroup
\end{itemize}
\end{example}

\subsection{Operational Specifications}\label{sec:op-spec}

Operational event/data specifications are introduced as a means to
specify in a more constructive style the properties of event/data
systems.  They are not appropriate for writing abstract requirements for
which axiomatic specifications are recommended.  Though $\EDHL$-logic is
able to specify concrete models, as discussed in \cref{sec:logic}, the
use of operational specifications allows a graphic representation close
to familiar formalisms in the literature, like UML protocol state
machines, cf.~\cite{uml-2.5,knapp-et-al:fase:2015}.  As will be shown in
\cref{sec:expressivity}, finite operational specifications can be
characterised by a sentence in $\EDHL$-logic.  Therefore, $\EDHL$-logic
is still the common basis of our development approach.  Transitions in
an operational specification are tuples $(c, \varphi, e, \psi, c')$ with
$c$ a source control state, $\varphi$ a precondition, $e$ an event,
$\psi$ a state transition predicate specifying the possible effects of
the event $e$, and $c'$ a target control state. In the semantic models
an event must be enabled whenever the respective source data state
satisfies the precondition.  Thus isolating preconditions has a semantic
consequence that is not expressible by transition predicates only.  The
effect of the event must respect $\psi$; no other transitions are
allowed.

\begin{definition}%[Operational specification]
An \emph{operational ed specification} $O = (\Sigma, C,\allowbreak
T,\allowbreak (c_0, \varphi_0))$ is given by an ed signature $\Sigma$, a
set of \emph{control states} $C$, a \emph{transition relation
  specification} $T \subseteq C \times \stFrm(\Sigma) \times
\Evt(\Sigma) \times \trFrm(\Sigma) \times C$, an \emph{initial control
  state} $c_0 \in C$, and an \emph{initial state predicate} $\varphi_0
\in \stFrm(\Sigma)$, such that $C$ is \emph{syntactically reachable},
%\ednote[RH]{Is the condition syntactically reachable anywhere needed? If
%  not, please drop it.}\ednote[AK]{Yes, it is needed for constructing an
%  $\EDHL$-sentence from an op ed spec.}
  i.e., for every $c \in C
\setminus \{ c_0 \}$ there are $(c_0, \varphi_1, e_1, \psi_1, c_1),
\ldots, (c_{n-1}, \varphi_n, e_n,\allowbreak \psi_n,\allowbreak c_n) \in
T$ with $n > 0$ such that $c_n = c$.  We write $\Sig(O)$ for $\Sigma$,
etc.

A $\Sigma$"=edts $M$ is a \emph{model} of $O$ if $\CtrlSt(M) = C$ up to
a bijective renaming, $\ctrlst_0(M) = c_0$, $\DataSt_0(M) \subseteq \{
\omega \mid \omega \datamodels{\Attr(\Sigma)} \varphi_0 \}$, and if the
following conditions hold:
\begin{itemize}[leftmargin=*, topsep=2pt]
  \item for all $(c, \varphi, e, \psi, c') \in \Trans$ and $\omega \in
\DataSt(\Attr(\Sigma))$ with $\omega \datamodels{\Attr(\Sigma)}
\varphi$, there is a $((c, \omega),\allowbreak (c',
\omega')) \in \Rel(M)_e$ with $(\omega,\allowbreak \omega')
\datamodels{\Attr(\Sigma)} \psi$;
  \item for all $((c, \omega), (c', \omega')) \in R(M)_e$ there is a
$(c, \varphi, e, \psi, c') \in \Trans$ with $\omega
\datamodels{\Attr(\Sigma)} \varphi$ and $(\omega,\allowbreak \omega')
\datamodels{\Attr(\Sigma)} \psi$.
\end{itemize}
The class of all models of $O$ is denoted by $\Mod(O)$.  The
\emph{semantics} of $O$ is given by the pair $(\Sig(O), \Mod(O))$ where
$\Sig(O) = \Sigma$.
\end{definition}

\begin{example}\label{ex:atm}
We construct an operational ed specification, called $\mathit{ATM}$, for
the ATM example.  The signature of $\mathit{ATM}$ extends the one of
$\Spec_1$ (and $\Spec_0$) by an additional integer-valued attribute
$\mathsf{trls}$ which counts the number of attempts to enter a correct
PIN (with the same card).  $\mathit{ATM}$ is graphically presented in
\cref{fig:op-spec-ATM}.  The initial control state is $\mathit{Card}$,
and the initial state predicate is $\truefrm$.  Preconditions are
written before the symbol $\limp$. If no precondition is explicitly
indicated it is assumed to be $\truefrm$.  Due to the extended
signature, $\mathit{ATM}$ is not a simple implementation of $\Spec_1$,
and we will only formally justify the implementation relationship in
\cref{ex:atm-justify}.
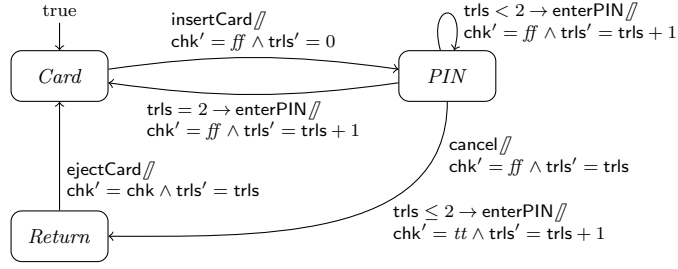
\begin{figure}[!t]\centering
\begin{tikzpicture}[scale=.85, transform shape]
\tikzset{
  every node/.append style={align=left},
  every state/.append style={rectangle, rounded corners, minimum width=1.5cm},
}
\node[state, initial text={\fontsize{8pt}{8pt}\selectfont$\truefrm$}, initial above] (Card) {$\mathit{Card}$};
\node[state, right of=Card, node distance=6.0cm] (PIN) {$\mathit{PIN}$};
\node[state, below of=Card, node distance=2.5cm] (Return) {$\mathit{Return}$};
\path[->,font={\fontsize{8pt}{8pt}\selectfont}] 
  (Card) edge[bend left=8] node[anchor=base, above] {$\mathsf{insertCard} \nfatslash{}$\\ $\mathsf{chk}' = \falseval \land \mathsf{trls}' = 0$} (PIN)
  (PIN) edge[loop above, min distance=8mm] node[anchor=west, right, xshift=4pt, yshift=-4pt] {$\mathsf{trls} < 2 \limp \mathsf{enterPIN} \nfatslash{}$\\ $\mathsf{chk}' = \falseval \land \mathsf{trls}' = \mathsf{trls}+1$ } (PIN)
  (PIN) edge[in=0, out=-90] node[anchor=east, pos=.5, right, xshift=24pt]{$\mathsf{trls} \leq 2 \limp \mathsf{enterPIN} \nfatslash{}$\\$\mathsf{chk}' = \trueval \land \mathsf{trls}' = \mathsf{trls}+1$}
                            node[anchor=east, pos=.15, right, xshift=3pt] {$\mathsf{cancel} \nfatslash{}$\\ $\mathsf{chk}' = \falseval \land \mathsf{trls}' = \mathsf{trls}$} (Return)
  (PIN) edge[bend left=8] node[anchor=top, below] {$\mathsf{trls} = 2 \limp \mathsf{enterPIN} \nfatslash{}$\\$\mathsf{chk}' = \falseval \land \mathsf{trls}' = \mathsf{trls}+1$} (Card)
  (Return) edge[] node[anchor=west, right, yshift=-10pt] {$\mathsf{ejectCard} \nfatslash{}$\\$\mathsf{chk}' = \mathsf{chk} \land \mathsf{trls}' = \mathsf{trls}$} (Card)
;
\end{tikzpicture}
\caption{Operational ed specification $\mathit{ATM}$}%
\label{fig:op-spec-ATM}%
\vskip-2pt
\end{figure}
\end{example}

Operational specifications can be composed by a syntactic parallel
composition operator which synchronises shared events.  Two ed
signatures $\Sigma_1$ and $\Sigma_2$ are \emph{composable} if
$\Attr(\Sigma_1) \cap \Attr(\Sigma_2) = \emptyset$.  Their parallel
composition is given by $\Sigma_1 \otimes \Sigma_2 = (\Evt(\Sigma_1)
\cup \Evt(\Sigma_2), \Attr(\Sigma_1) \cup \Attr(\Sigma_2))$.

\begin{definition}
Let $\Sigma_1$ and $\Sigma_2$ be composable ed signatures and let $O_1$
and $O_2$ be operational ed specifications with $\Sigma(O_1) = \Sigma_1$
and $\Sigma(O_2) = \Sigma_2$.  The \emph{parallel composition} of $O_1$
and $O_2$ is given by the operational ed specification $O_1 \parallel
O_2 = (\Sigma_1 \otimes \Sigma_2,\allowbreak C,\allowbreak T,\allowbreak (c_0, \varphi_0))$ with $c_0 =
(\ctrlst_0(O_1), \ctrlst_0(O_2))$, $\varphi_0 = \initpred(O_1) \land
\initpred(O_2)$, and $C$ and $T$ are inductively defined by $c_0 \in C$
and
\begin{itemize}[leftmargin=*, topsep=2pt]
  \item for $e_1 \in \Evt(\Sigma_1) \setminus \Evt(\Sigma_2)$, $c_1,
c_1' \in \CtrlSt(O_1)$, and $c_2 \in \CtrlSt(O_2)$, if $(c_1, c_2) \in
C$ and $(c_1, \varphi_1, e_1,\allowbreak \psi_1,\allowbreak c_1') \in
\Trans(O_1)$, then $(c_1', c_2) \in C$ and $((c_1, c_2), \varphi_1, e_1,
\psi_1 \land \id{A(\Sig_2)},\allowbreak (c_1', c_2)) \in T$;

  \item for $e_2 \in \Evt(\Sigma_2) \setminus \Evt(\Sigma_1)$, $c_2,
c_2' \in \CtrlSt(O_2)$, and $c_1 \in \CtrlSt(O_1)$, if $(c_1, c_2) \in
C$ and $(c_2, \varphi_2, e_2,\allowbreak \psi_2,\allowbreak c_2') \in
\Trans(O_2)$, then $(c_1, c_2') \in C$ and $((c_1, c_2), \varphi_2, e_2,
\psi_2 \land \id{A(\Sig_1)},\allowbreak (c_1, c_2')) \in T$;

  \item for $e \in \Evt(\Sigma_1) \cap \Evt(\Sigma_2)$, $c_1, c_1' \in
\CtrlSt(O_1)$, and $c_2, c_2' \in \CtrlSt(O_2)$, if $(c_1, c_2) \in C$,
$(c_1, \varphi_1, e, \psi_1, c_1') \in \Trans(O_1)$, and $(c_2,
\varphi_2, e,\allowbreak \psi_2,\allowbreak c_2') \in \Trans(O_2)$, then
$(c_1', c_2') \in C$ and $((c_1, c_2), \varphi_1 \land \varphi_2, e,
\psi_1 \land \psi_2, (c_1', c_2')) \in T$.\footnote{
Note that joint moves with $e$ cannot become inconsistent due to composability of ed signatures.}
\end{itemize}
\end{definition}

\subsection{Expressiveness of $\EDHL$-Logic}\label{sec:expressivity}

We show that the semantics of an operational ed specification $O$ with
finitely many control states can be characterised by a single
\EDHL-sentence $\varrho_O$, i.e., an edts $M$ is a model of $O$ iff $M
\EDHLmodels{\Sig(O)} \varrho_O$.
\begin{algorithm}[t]
\caption{Constructing a sentence from an operational ed specification}
\label{alg:sen-op-spec}
\begin{algorithmic}[1]
\Require $O \equiv \text{finite operational ed specification}$
\Requirx $\mathit{Im}_O(c) = \{ (\varphi, e, \psi, c') \mid (c, \varphi, e, \psi, c') \in T(O) \}$ \text{for $c \in C(O)$}
\Requirx $\mathit{Im}_O(c, e) = \{ (\varphi, \psi, c') \mid (c, \varphi, e, \psi, c') \in T(O) \}$ \text{for $c \in C(O)$, $e \in \Evt(\Sig(O))$}
\Statex
\Function {$\mathrm{sen}$}{$c, I, V, B$} \Comment{$c$: state, $I$: image to visit, $V$: states to visit, $B$: bound states}
  \If{$I \neq \emptyset$}
    \State $(\varphi, e, \psi, c') \gets \Choose\ I$
    \If{$c' \in B$}
      \State \Return $\hyat c \,.\, \varphi \limp \dldia{\lab{e}{\psi}}(c' \land \mathrm{sen}(c, I \setminus \{ (\varphi, e, \psi, c') \}, V, B))$
    \Else
      \State \Return $\hyat c \,.\, \varphi \limp \dldia{\lab{e}{\psi}}(\hybind c' \,.\, \mathrm{sen}(c, I \setminus \{ (\varphi, e, \psi, c') \}, V, B \cup \{ c' \}))$
    \EndIf
  \EndIf
  \State $V \gets V \setminus \{ c \}$
  \If{$V \neq \emptyset$}
    \State $c' \gets \Choose\ B \cap V$
    \State \Return $\mathrm{fin}(c) \land \mathrm{sen}(c', \mathit{Im}_O(c'), V, B)$
  \EndIf
  \State \Return $\mathrm{fin}(c) \land \bigwedge_{c_1 \in C(O), c_2 \in C(O) \setminus \{ c_1 \}} \neg\hyat c_1 \,.\, c_2$
\EndFunction
\Statex
\Function {$\mathrm{fin}$}{$c$}
\State \Return $\hyat c \,.\, \begin{array}[t]{@{}l@{}}
                                \bigwedge_{e \in \Evt(\Sig(O))} \bigwedge_{P \subseteq \mathit{Im}_O(c, e)}\\{}
                                [e\nfatslash\begin{array}[t]{@{}l@{}}
                                              \big(\bigwedge_{(\varphi, \psi, c') \in P} (\varphi \land \psi)\big) \land{}\\
                                              \neg\big(\bigvee_{(\varphi, \psi, c') \in \mathit{Im}_O(c, e) \setminus P} (\varphi \land \psi)\big)
                                ] \big(\bigvee_{(\varphi, \psi, c') \in P} c'\big)
                                            \end{array}
                              \end{array}$
\EndFunction
\end{algorithmic}
\end{algorithm}
Using \cref{alg:sen-op-spec}, such a characterising sentence is
%
%\bgroup\abovedisplayskip6pt\belowdisplayskip6pt
\begin{equation*}
  \varrho_O
=
  \hybind c_0 \,.\, \varphi_0 \land \mathrm{sen}(c_0, \mathit{Im}_O(c_0), \CtrlSt(O), \{ c_0 \})
\ \text{,}
\end{equation*}
%\egroup
%
where $c_0 = \ctrlst_0(O)$ and $\varphi_0 = \initpred(O)$.
\Cref{alg:sen-op-spec} closely follows the procedure
in~\cite{madeira-et-al:tcs:2018} for characterising a finite structure by
a sentence of $\DHL$-logic.  A call $\mathrm{sen}(c,\allowbreak I, V,
B)$ performs a recursive breadth-first traversal through $O$ starting
from $c$, where $I$ holds the unprocessed quadruples $(\varphi, e, \psi,
c')$ of transitions outgoing from $c$, $V$ the remaining states to
visit, and $B$ the set of already bound states.  The function first
requires the existence of each outgoing transition of $I$, provided its
precondition holds, in the resulting formula, binding any newly reached
state.  Then it requires that no other transitions with source state $c$
exist using calls to $\mathrm{fin}$.  Having visited all states in $V$,
it finally requires all states in $\CtrlSt(O)$ to be pairwise different.

It is $\mathrm{fin}(c)$ where this algorithm mainly deviates
from~\cite{madeira-et-al:tcs:2018}: To ensure that no other transitions
from $c$ exist than those specified in $O$, $\mathrm{fin}(c)$ produces
the requirement that at state $c$, for every event $e$ and for every
subset $P$ of the transitions outgoing from $c$, whenever an
$e$-transition can be done with the combined effect of $P$ but not
adhering to any of the effects of the currently not selected
transitions, the $e$-transition must have one of the states as its
target that are target states of $P$.  The rather complicated
formulation is due to possibly overlapping preconditions where for a
single event $e$ the preconditions of two different transitions may be
satisfied simultaneously.  For a state $c$, where all outgoing
transitions for the same event have disjoint preconditions, the
$\EDHL$-formula returned by $\mathrm{fin}(c)$ is equivalent to
%
%\bgroup\abovedisplayskip6pt\belowdisplayskip-1pt
\begin{equation*}\textstyle
  \hyat c \,.\, \bigwedge_{e \in \Evt(\Sig(O))} \begin{array}[t]{@{}l@{}}
                                                \bigwedge_{(\varphi, \psi, c') \in \mathit{Im}_O(c, e)} \dlbox{\lab{e}{\varphi \land \psi}} c' \land{}\\
                                                \dlbox{\lab{e}{\neg\big(\bigvee_{(\varphi, \psi, c') \in \mathit{Im}_O(c, e)} (\varphi \land \psi)\big)}} \falsefrm\ \text{.}
                                              \end{array}
\end{equation*}
%\egroup

\begin{example}
We show the first few steps of representing the operational ed
specification $\mathit{ATM}$ of \cref{fig:op-spec-ATM} as an
$\EDHL$-sentence $\varrho_{\mathit{ATM}}$.  This top-level sentence is
\bgroup\abovedisplayskip6pt\belowdisplayskip6pt
\begin{equation*}
  \hybind \mathit{Card} \,.\, \truefrm \land \mathrm{sen}(\begin{array}[t]{@{}l@{}}
                                                            \mathit{Card}, \{ (\truefrm, \mathsf{insertCard}, \mathsf{chk}' = \falseval \land \mathsf{trls}' = 0, \mathit{PIN}) \},\\
                                                            \{ \mathit{Card}, \mathit{PIN}, \mathit{Return} \}, \{ \mathit{Card} \})\ \text{.}
                                                          \end{array}
\end{equation*}
\egroup
The first call of $\mathrm{sen}(\mathit{Card}, \ldots)$ explores the
single outgoing transition from $\mathit{Card}$ to $\mathit{PIN}$, adds
$\mathit{PIN}$ to the bound states, and hence expands to
\bgroup\abovedisplayskip6pt\belowdisplayskip6pt
\begin{equation*}
  \hyat \mathit{Card} \,.\, \truefrm \limp \begin{array}[t]{@{}l@{}}
                                             \dldia{\lab{\mathsf{insertCard}}{\mathsf{chk}' = \falseval \land \mathsf{trls}' = 0}}\hybind \mathit{PIN} \,.\, \\
                                             \qquad\mathrm{sen}(\mathit{Card}, \emptyset, \{ \mathit{Card}, \mathit{PIN}, \mathit{Return} \}, \{ \mathit{Card}, \mathit{PIN} \})\ \text{.}
                                           \end{array}
\end{equation*}
\egroup
Now all outgoing transitions from $\mathit{Card}$ have been explored and
the next call of $\mathrm{sen}(\mathit{Card},\allowbreak \emptyset, \ldots)$
removes $\mathit{Card}$ from the set of states to be visited, resulting in
\bgroup\abovedisplayskip6pt\belowdisplayskip6pt
\begin{equation*}
  \mathrm{fin}(\mathit{Card}) \land \mathrm{sen}(\begin{array}[t]{@{}l@{}}
                                                   \mathit{PIN}, \{ \begin{array}[t]{@{}l@{}}
                                                                      (\mathsf{trls < 2}, \mathsf{enterPIN}, \ldots),
                                                                      (\mathsf{trls = 2}, \mathsf{enterPIN}, \ldots),\\
                                                                      (\mathsf{trls \leq 2}, \mathsf{enterPIN}, \ldots),
                                                                      (\truefrm, \mathsf{cancel}, \ldots) \},
                                                                    \end{array}\\
                                                   \{ \mathit{PIN}, \mathit{Return} \}, \{ \mathit{Card}, \mathit{PIN} \})\ \text{.}
                                                 \end{array}
\end{equation*}
\egroup
As there is only a single outgoing transition from $\mathit{Card}$, the special case of disjoint preconditions applies for the finalisation call, and $\mathrm{fin}(\mathit{Card})$ results in
\bgroup\abovedisplayskip6pt\belowdisplayskip6pt
\begin{equation*}
  \hyat \mathit{Card} \,.\, \begin{array}[t]{@{}l@{}}
                              \dlbox{\lab{\mathsf{insertCard}}{\mathsf{chk}' = \falseval \land \mathsf{trls}' = 0}} \mathit{PIN} \land{}\\
                              \dlbox{\lab{\mathsf{insertCard}}{\mathsf{chk}' = \trueval \lor \mathsf{trls}' \neq 0}}\falsefrm \land{}\\
                              \dlbox{\lab{\mathsf{enterPIN}}{\truefrm}}\falsefrm \land
                              \dlbox{\lab{\mathsf{cancel}}{\truefrm}}\falsefrm \land
                              \dlbox{\lab{\mathsf{ejectCard}}{\truefrm}}\falsefrm\ \text{.}
                            \end{array}
\end{equation*}
\egroup
\end{example}

%%% Local Variables:
%%% mode: LaTeX
%%% mode: TeX-PDF
%%% mode: TeX-source-correlate
%%% TeX-master: "edts-hdl.tex"
%%% End:

%!TEX root = edts-hdl.tex

\section{Constructor Implementations}\label{sec:constructor-impl}

The implementation notion defined in \cref{sec:axiomatic} is too simple
for many practical applications.  It requires the same signature for
specification and implementation and does not support the process of
constructing an implementation.  Therefore, Sannella and
Tarlecki~\cite{DBLP:journals/acta/SannellaT88,sannella-tarlecki:2012}
have proposed the notion of constructor implementation which is a
generic notion applicable to specification formalisms which are based on
signatures and semantic structures for signatures.  We will reuse the
ideas in the context of $\EDHL$"=logic.

The notion of \emph{constructor} is the basis: for signatures $\Sigma_1,
\ldots,\Sigma_n, \Sigma \in \EDHLSig$, a \emph{constructor} $\kappa$
from $(\Sigma_1, \ldots, \Sigma_n)$ to $\Sigma$ is a (total) function
$\kappa : \EDHLStr(\Sigma_1) \times \ldots \times \EDHLStr(\Sigma_n) \to
\EDHLStr(\Sigma)$.  Given a constructor $\kappa$ from $(\Sigma_1,
\ldots, \Sigma_n)$ to $\Sigma$ and a set of constructors $\kappa_i$ from
$(\Sigma_i^1, \ldots, \Sigma_i^{k_i})$ to $\Sigma_i$, $1\leq i \leq n$,
the constructor $(\kappa_1, \ldots, \kappa_n);\kappa$ from $(\Sigma_1^1,
\ldots, \Sigma_1^{k_1}, \ldots, \Sigma_n^1, \ldots, \Sigma_n^{k_n})$ to
$\Sigma$ is obtained by the usual composition of functions.  The
following definitions apply to both axiomatic and operational ed
specifications since the semantics of both is given in terms of ed
signatures and model classes of edts.  In particular, the implementation
notion allows to implement axiomatic specifications by operational
specifications.

\begin{definition}%[Constructor implementation]
\label{def:constructor-impl}
Given specifications $\Spec, \Spec_1, \ldots,\allowbreak \Spec_n$ and a
constructor $\kappa$ from $(\Sig(\Spec_1),\allowbreak \ldots,\allowbreak
\Sig(\Spec_n))$ to $\Sig(\Spec)$, the tuple $\langle\Spec_1, \ldots,
\Spec_n\rangle$ is a \emph{constructor implementation via} $\kappa$ of
$\Spec$, in symbols $\Spec \refinesto_\kappa \langle\Spec_1, \dots, \Spec_n\rangle$,
%
%\begin{equation*}
%  \Spec \refinesto_\kappa \langle\Spec_1, \dots, \Spec_n\rangle
%\end{equation*}
%
if for all $M_i \in \Mod(\Spec_i)$ we have $\kappa(M_1, \ldots ,M_n) \in
\Mod(\Spec).$ The implementation involves a
\emph{decomposition} if $n > 1$.
\end{definition}

The notion of simple implementation in \cref{sec:axiomatic} is captured
by choosing the identity.  We now introduce a set of more advanced
constructors in the context of ed signatures and edts.  Let us first
consider two central notions for constructors: signature morphisms and
reducts. For data signatures $A, A'$ a \emph{data signature morphism}
$\sigma : A \to A'$ is a function from $A$ to $A'$.  The
$\sigma$"=\emph{reduct} of an $A'$"=data state $\omega' : A' \to \Data$
is given by the $A$"=data state $\reduct{\omega'}{\sigma} : A \to \Data$
defined by $(\reduct{\omega'}{\sigma})(a) = \omega'(\sigma(a))$ for every
$a \in A$.  If $A \subseteq A'$, the injection of $A$ into $A'$ is a
particular data signature morphism and we denote the reduct of an
$A'$"=data state $\omega'$ to $A$ by $\restrict{\omega'}{A}$.  If $A =
A_1 \cup A_2$ is the disjoint union of $A_1$ and $A_2$ and $\omega_i$
are $A_i$"=data states for $i \in \{ 1, 2 \}$ then $\omega_1+\omega_2$
denotes the unique $A$"=data state $\omega$ with $\restrict{\omega}{A_i}
= \omega_i$ for $i \in \{ 1, 2 \}$.  The $\sigma$-reduct
$\reduct{\gamma}{\sigma}$ of a configuration $\gamma = (c, \omega')$ is
given by $(c, \reduct{\omega'}{\sigma})$, and is lifted to a set of
configurations $\Gamma'$ by $\reduct{\Gamma'}{\sigma} = \{
\reduct{\gamma'}{\sigma} \mid \gamma' \in \Gamma' \}$.

\begin{definition}%[Signature morphism]
An \emph{ed signature morphism} $\sigma = (\sigma_{\Evt},
\sigma_{\Attr}) : \Sigma \to \Sigma'$ is given by a function
$\sigma_{\Evt} : \Evt(\Sigma) \to \Evt(\Sigma')$ and a data signature
morphism $\sigma_{\Attr} : \Attr(\Sigma) \to \Attr(\Sigma')$.  We
abbreviate both $\sigma_{\Evt}$ and $\sigma_{\Attr}$ by $\sigma$.
\end{definition}

\begin{definition}%[Reduct]
Let $\sigma : \Sigma \to \Sigma'$ be an ed signature morphism and $M'$ a
$\Sigma'$-edts.  The $\sigma$-\emph{reduct} of $M'$ is the $\Sigma$-edts
$\reduct{M'}{\sigma} = (\Gamma, R, \Gamma_0)$ such that $\Gamma_0 =
\reduct{\Conf_0(M')}{\sigma}$, and $\Gamma$ and $R = (R_e)_{e \in
  E(\Sigma)}$ are inductively defined by $\Gamma_0 \subseteq \Gamma$ and
for all $e \in \Evt(\Sigma)$, $\gamma', \gamma'' \in \Conf(M')$:
 if
$\reduct{\gamma'}{\sigma} \in \Gamma$ and $(\gamma', \gamma'') \in
\Rel(M')_{\sigma(e)}$, then $\reduct{\gamma''}{\sigma} \in \Gamma$ and
$(\reduct{\gamma'}{\sigma}, \reduct{\gamma''}{\sigma}) \in R_e$.
\end{definition}

\begin{definition}%[Reduct, relabelling, and restriction]
\label{def:alphabet-extension}
Let $\sigma : \Sigma \to \Sigma'$ be an ed signature morphism.  The
\emph{reduct constructor} $\kappa_{\sigma}$ from $\Sigma'$ to $\Sigma$
maps any $M' \in \EDHLStr(\Sigma')$ to its reduct $\kappa_{\sigma}(M') =
\reduct{M'}{\sigma}$. Whenever $\sigma_{\Attr}$ and $\sigma_{\Evt}$ are
bijective functions, $\kappa_\sigma$ is a \emph{ relabelling
  constructor}. If $\sigma_{\Evt}$ and $\sigma_{\Attr}$ are injective,
$\kappa_\sigma$ is a \emph{restriction constructor}.
\end{definition}

\begin{example}\label{ex:atm-justify}
The operational specification $\mathit{ATM}$ is a constructor implementation
of $\Spec_1$ via the restriction constructor $\kappa_{\iota}$ determined by the inclusion
signature morphism $\iota: \Sig(\Spec_1) \to \Sig(\mathit{ATM})$, i.e.,
$\Spec_1 \refinesto_{\kappa_\iota} \mathit{ATM}$.
\end{example}

A further refinement technique for reactive systems (see, e.g.,
\cite{actionrefinement}), is the implementation of simple events by
complex events, like their sequential composition.  To formalise this as
a constructor we use \emph{composite events} $\Ecomp(E)$ over a given
set of events $E$, given by the grammar $\theta ::= e \mid \theta +
\theta \mid \theta ;\theta \mid \theta^*$ with $e \in E$.  They are
\emph{interpreted} over an $(E, A)$-edts $M$ by $\Rel(M)_{\theta_1 +
  \theta_2} = \Rel(M)_{\theta_1} \cup \Rel(M)_{\theta_2}$,
$\Rel(M)_{\theta_1; \theta_2} = \Rel(M)_{\theta_1}; \Rel(M)_{\theta_2}$,
and $\Rel(M)_{\theta^*} = (\Rel(M)_{\theta})^*$.  Then we can introduce
the intended constructor by means of reducts over signature morphisms
mapping atomic to composite events:

\begin{definition}%[Event refinement]
\label{def:action-ref}
Let $\Sigma, \Sigma'$ be ed signatures, $D'$ a finite subset of
$\Ecomp(\Evt(\Sigma'))$, $\Delta' = (D',\allowbreak \Attr(\Sigma'))$,
and $\alpha : \Sigma \to \Delta'$ an ed signature morphism.  The
\emph{event refinement} constructor $\kappa_\alpha$
%$\reductop\alpha$
from $\Delta'$ to
$\Sigma$ maps any $M' \in \EDHLStr(\Delta')$ to its reduct
$\reduct{M'}{\alpha} \in \EDHLStr(\Sigma)$.
\end{definition}

Finally, we consider a semantic, synchronous parallel composition
constructor that allows for decomposition of implementations into
components which synchronise on shared events.  Given two composable
signatures $\Sigma_1$ and $\Sigma_2$, the \emph{parallel composition}
$\gamma_1 \otimes \gamma_2$ of two configurations $\gamma_1 = (c_1,
\omega_1)$, $\gamma_2 = (c_2, \omega_2)$ with $\omega_1 \in
\DataSt(\Attr(\Sigma_1))$, $\omega_2 \in \DataSt(\Attr(\Sigma_2))$ is
given by $((c_1, c_2), \omega_1 + \omega_2)$, and lifted to two sets of
configurations $\Gamma_1$ and $\Gamma_2$ by $\Gamma_1 \otimes \Gamma_2 =
\{ \gamma_1 \otimes \gamma_2 \mid \gamma_1 \in \Gamma_1,\ \gamma_2 \in
\Gamma_2 \}$.

\begin{definition}%[Parallel composition]
Let $\Sigma_1, \Sigma_2$ be composable ed signatures.  The
\emph{parallel composition constructor} $\kappa_\otimes$ from
$(\Sigma_1, \Sigma_2)$ to $\Sigma_1 \otimes \Sigma_2$ maps any $M_1 \in
\EDHLStr(\Sigma_1)$, $M_2 \in \EDHLStr(\Sigma_2)$ to $M_1 \otimes M_2 =
(\Gamma, R, \Gamma_0) \in \EDHLStr(\Sigma_1 \otimes \Sigma_2)$, where
$\Gamma_0 = \Conf_0(M_1) \otimes \Conf_0(M_2)$, and $\Gamma$ and $R =
(R_e)_{\Evt(\Sigma_1) \cup \Evt(\Sigma_2)}$ are inductively defined by
$\Gamma_0 \subseteq \Gamma$ and
\begin{itemize}[leftmargin=*, topsep=2pt]
  \item for all $e_1 \in \Evt(\Sigma_1) \setminus \Evt(\Sigma_2)$,
$\gamma_1, \gamma_1' \in \Conf(M_1)$, and $\gamma_2 \in \Conf(M_2)$, if
$\gamma_1 \otimes \gamma_2 \in \Gamma$ and $(\gamma_1, \gamma_1') \in
\Rel(M_1)_{e_1}$, then $\gamma_1' \otimes \gamma_2 \in \Gamma$ and
$(\gamma_1 \otimes \gamma_2, \gamma_1' \otimes \gamma_2) \in R_{e_1}$;

  \item for all $e_2 \in \Evt(\Sigma_2) \setminus \Evt(\Sigma_1)$,
$\gamma_2, \gamma_2' \in \Conf(M_2)$, and $\gamma_1 \in \Conf(M_1)$, if
$\gamma_1 \otimes \gamma_2 \in \Gamma$ and $(\gamma_2, \gamma_2') \in
\Rel(M_2)_{e_2}$, then $\gamma_1 \otimes \gamma_2' \in \Gamma$ and
$(\gamma_1 \otimes \gamma_2, \gamma_1 \otimes \gamma_2') \in R_{e_2}$;

  \item for all $e \in \Evt(\Sigma_1) \cap \Evt(\Sigma_2)$, $\gamma_1,
\gamma_1' \in \Conf(M_1)$, and $\gamma_2, \gamma_2' \in \Conf(M_2)$,
if $\gamma_1 \otimes \gamma_2 \in \Gamma$, $(\gamma_1, \gamma_1') \in
\Rel(M_1)_{e_1}$, and $(\gamma_2, \gamma_2') \in \Rel(M_2)_{e_2}$, then
$\gamma_1' \otimes \gamma_2' \in \Gamma$ and $(\gamma_1 \otimes
\gamma_2, \gamma_1' \otimes \gamma_2') \in R_{e}$.
\end{itemize}
\end{definition}

An obvious question is how the semantic parallel composition constructor is
related to the syntactic parallel composition of operational ed specifications.

\begin{proposition}\label{prop:op-ed-spec-parallel}
Let $O_1, O_2$ be operational ed specifications with composable signatures. 
Then $\Mod(O_1) \otimes \Mod(O_2)
\subseteq \Mod(O_1 \parallel O_2)$,
where $\Mod(O_1) \otimes \Mod(O_2)$ denotes $\kappa_\otimes(\Mod(O_1),\Mod(O_2))$.
\end{proposition}

The converse $\Mod(O_1 \parallel O_2) \subseteq \Mod(O_1) \otimes
\Mod(O_2)$ does not hold: Consider the ed signature $\Sigma = (E, A)$
with $E = \{ e \}$, $A = \emptyset$, and the operational ed
specifications $O_i = (\Sigma,\allowbreak C_i,\allowbreak
T_i,\allowbreak (c_{i, 0},\allowbreak \varphi_{i, 0}))$ for $i \in \{ 1,
2 \}$ with $C_1 = \{ c_{1, 0} \}$, $T_1 = \{ (c_{1, 0},\allowbreak
\truefrm,\allowbreak e, \falsefrm,\allowbreak c_{1, 0}) \}$,
$\varphi_{1, 0} = \truefrm$; and $C_2 = \{ c_{2, 0} \}$, $T_2 =
\emptyset$, $\varphi_{2, 0} = \truefrm$.  Then $\Mod(O_1) = \emptyset$,
but $\Mod(O_1 \parallel O_2) = \{ M \}$ with $M$ showing just the
initial configuration.

The next theorem shows the usefulness of the syntactic parallel
composition operator for proving implementation correctness when a
(semantic) parallel composition constructor is involved. The theorem is
a direct consequence of~\cref{prop:op-ed-spec-parallel}
and~\cref{def:constructor-impl}.

\begin{theorem}\label{thm:par}
Let $\Spec$ be an (axiomatic or operational) ed specification, $O_1, O_2$
operational ed specifications with composable signatures, and $\kappa$
an implementation constructor from $\Sig(O_1) \otimes \Sig(O_2)$ to
$\Sig(\Spec)$: If $\Spec \refinesto_\kappa O_1 \parallel O_2$, then
$\Spec \refinesto_{\kappa_\otimes; \kappa} \langle O_1, O_2 \rangle$.
\end{theorem}

\begin{figure}[!htp]\centering
\subfloat[Operational ed specification $\mathit{ATM}'$]{%
\label{fig:op-spec-ATM-ref}%
\begin{tikzpicture}[scale=0.85, transform shape]%shorten >=0pt]
\tikzset{
  every node/.append style={align=left},
  every state/.append style={rectangle, rounded corners, minimum width=1.5cm},
}
\node[state, initial text={\fontsize{8pt}{8pt}\selectfont$\truefrm$}, initial left] (Card) {$\mathit{Card}$};
\node[state, right of=Card, node distance=6.0cm] (PIN) {$\mathit{PIN}$};
\node[state] at (1, -2.0) (Return) {$\mathit{Return}$};
\node[state, below of=Card, node distance=5.0cm] (Verifying) {$\mathit{Verifying}$};
\node[state, right of=Verifying, node distance=6.0cm] (PINEntered) {$\mathit{PINEntered}$};
\path[->,font={\fontsize{8pt}{8pt}\selectfont}] 
  (Card) edge[] node[anchor=base, above] {$\mathsf{insertCard} \nfatslash{}$\\ $\mathsf{chk}' = \falseval \land \mathsf{trls}' = 0$} (PIN)
% (PIN) edge[in=60, out=200] node[below, yshift=-3pt] {$\mathsf{cancel} \nfatslash{}$\\ $\mathsf{chk}' = \falseval$} (Return)
  (PIN) edge[bend left=20] node[anchor=south east, pos=.25, yshift=0pt] {$\mathsf{cancel} \nfatslash{}$\\ $\mathsf{chk}' = \falseval \land $\\$\mathsf{trls}' = \mathsf{trls}$} (Return)
  (PIN) edge node[right] {$\mathsf{trls} \leq 2 \limp{}$\\ $\mathsf{enterPIN} \nfatslash{}$\\ $\mathsf{chk}' = \mathsf{chk} \land{}$\\ $\mathsf{trls}' = \mathsf{trls}$} (PINEntered)
  (PINEntered) edge node[below] {$\mathsf{trls} \leq 2 \limp \mathsf{verifyPIN} \nfatslash{}$\\ $\mathsf{chk}' = \mathsf{chk} \land \mathsf{trls}' = \mathsf{trls}$} (Verifying)
  (Verifying) edge[bend right=25] node[anchor=north west, pos=.5] {$\mathsf{trls} < 2 \limp{}$\\ $\mathsf{wrongPIN} \nfatslash{}$\\ $\mathsf{chk}' = \falseval \land{}$\\ $\mathsf{trls}' = \mathsf{trls}+1$} (PIN)
  (Verifying) edge node[anchor=north west, pos=.9] {$\mathsf{trls} \leq 2 \limp{}$\\ $\mathsf{correctPIN} \nfatslash{}$\\ $\mathsf{chk}' = \trueval \land{}$\\ $\mathsf{trls}' = \mathsf{trls}+1$} (Return)
  (Verifying) edge[bend left=10] node[anchor=east] {$\mathsf{trls} = 2 \limp{}$\\ $\mathsf{wrongPIN} \nfatslash{}$\\ $\mathsf{chk}' = \falseval \land{}$\\ $\mathsf{trls}' = \mathsf{trls}+1$} (Card)
  (Return) edge[] node[anchor=south west, pos=.2] {$\mathsf{ejectCard} \nfatslash{}$\\ $\mathsf{chk}' = \mathsf{chk} \land{}$\\ $\mathsf{trls}' = \mathsf{trls}$} (Card)
;
\end{tikzpicture}}

\subfloat[Operational specification $\mathit{CC}$ of a clearing company]{%
\label{fig:op-spec-CC}%
\begin{tikzpicture}[scale=.85, transform shape]
\tikzset{
  every node/.append style={align=left},
  every state/.append style={rectangle, rounded corners, minimum width=1.5cm},
}
\node[state, initial text={\fontsize{8pt}{8pt}\selectfont$\mathsf{cnt} = 0$}, initial left] (Idle) {$\mathit{Idle}$};
\node[state, right of=Idle, node distance=6.0cm] (Busy) {$\mathit{Busy}$};
\path[->,font={\fontsize{8pt}{8pt}\selectfont}] 
  (Idle) edge[bend left=15] node[anchor=base, above] {$\mathsf{verifyPIN} \nfatslash \mathsf{cnt}' = \mathsf{cnt}$} (Busy)
  (Busy) edge[bend left=4] node[anchor=base, above] {$\mathsf{correctPIN} \nfatslash \mathsf{cnt}' = \mathsf{cnt}+1$} (Idle)
  (Busy) edge[bend left=12] node[anchor=top, below] {$\mathsf{wrongPIN} \nfatslash \mathsf{cnt}' = \mathsf{cnt}+1$} (Idle)
;
\end{tikzpicture}}

\subfloat[Syntactic parallel composition $\mathit{ATM}' \parallel \mathit{CC}$]{%
\label{fig:op-spec-ATM-CC-comp}%
\begin{tikzpicture}[scale=0.85, transform shape]%shorten >=0pt]
\tikzset{
  every node/.append style={align=left},
  every state/.append style={rectangle, rounded corners, minimum width=1.5cm},
}
\node[state, initial text={\fontsize{8pt}{8pt}\selectfont$\mathsf{cnt} = 0$}, initial left] (CardIdle) {$\mathit{Card}, \mathit{Idle}$};
\node[state, right of=CardIdle, node distance=7.5cm] (PINIdle) {$\mathit{PIN}, \mathit{Idle}$};
\node[state] at (1.3, -2.5) (ReturnIdle) {$\mathit{Return}, \mathit{Idle}$};
\node[state, below of=CardIdle, node distance=5.5cm] (VerifyingBusy) {$\mathit{Verifying}, \mathit{Busy}$};
\node[state, right of=VerifyingBusy, node distance=7.5cm] (PINEnteredIdle) {$\mathit{PINEntered}, \mathit{Idle}$};
\path[->,font={\fontsize{8pt}{8pt}\selectfont}] 
  (CardIdle) edge[] node[anchor=base, above] {$\mathsf{insertCard} \nfatslash{}$\\ $\mathsf{chk}' = \falseval \land \mathsf{trls}' = 0 \land \mathsf{cnt} = \mathsf{cnt}'$} (PINIdle)
  (PINIdle) edge[bend left=20] node[anchor=south east, pos=.3, yshift=0pt] {$\mathsf{cancel} \nfatslash{}$\\ $\mathsf{chk}' = \falseval \land $\\$\mathsf{trls}' = \mathsf{trls}$\\ $\mathsf{cnt}' = \mathsf{cnt}$} (ReturnIdle)
  (PINIdle) edge node[right] {$\mathsf{trls} \leq 2 \limp{}$\\ $\mathsf{enterPIN} \nfatslash{}$\\ $\mathsf{chk}' = \mathsf{chk} \land{}$\\ $\mathsf{trls}' = \mathsf{trls} \land{}$\\ $\mathsf{cnt}' = \mathsf{cnt}$} (PINEnteredIdle)
  (PINEnteredIdle) edge node[below] {$\mathsf{trls} \leq 2 \limp \mathsf{verifyPIN} \nfatslash{}$\\ $\mathsf{chk}' = \mathsf{chk} \land \mathsf{trls}' = \mathsf{trls} \land \mathsf{cnt} = \mathsf{cnt}'$} (VerifyingBusy)
  (VerifyingBusy) edge[bend right=25] node[anchor=north west, pos=.52, yshift=2pt] {$\mathsf{trls} < 2 \limp{}$\\ $\mathsf{wrongPIN} \nfatslash{}$\\ $\mathsf{chk}' = \falseval \land{}$\\ $\mathsf{trls}' = \mathsf{trls}+1 \land{}$\\ $\mathsf{cnt}' = \mathsf{cnt}+1$} (PINIdle)
  (VerifyingBusy) edge node[anchor=north west, pos=.94, xshift=-1pt] {$\mathsf{trls} \leq 2 \limp{}$\\ $\mathsf{correctPIN} \nfatslash{}$\\ $\mathsf{chk}' = \trueval \land{}$\\ $\mathsf{trls}' = \mathsf{trls}+1 \land{}$\\ $\mathsf{cnt}' = \mathsf{cnt}+1$} (ReturnIdle)
  (VerifyingBusy) edge[bend left=10] node[anchor=east] {$\mathsf{trls} = 2 \limp{}$\\ $\mathsf{wrongPIN} \nfatslash{}$\\ $\mathsf{chk}' = \falseval \land{}$\\ $\mathsf{trls}' = \mathsf{trls}+1 \land{}$\\ $\mathsf{cnt}' = \mathsf{cnt}+1$} (CardIdle)
  (ReturnIdle) edge[] node[anchor=south west, pos=.2] {$\mathsf{ejectCard} \nfatslash{}$\\ $\mathsf{chk}' = \mathsf{chk} \land{}$\\ $\mathsf{trls}' = \mathsf{trls} \land{}$\\ $\mathsf{cnt}' = \mathsf{cnt}$} (CardIdle)
;
\end{tikzpicture}}
\caption{Operational ed specifications $\mathit{ATM}'$, $\mathit{CC}$
  and their parallel composition}
\end{figure}
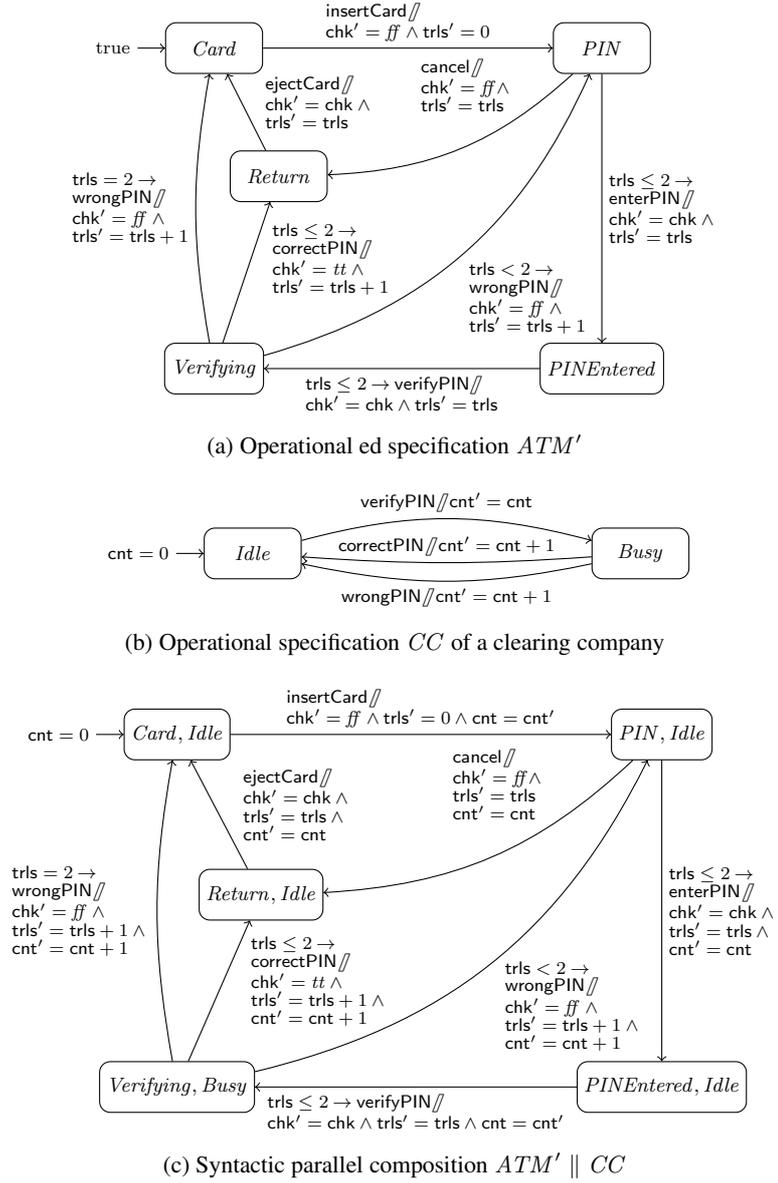
\begin{example}
We finish the refinement chain for the ATM specifications by applying a
decomposition into two parallel components.  The operational
specification $\mathit{ATM}$ of \cref{ex:atm} (and
\cref{ex:atm-justify}) describes the interface behaviour of an ATM
interacting with a user.  For a concrete realisation, however, an ATM
will also interact internally with other components, like, e.g., a
clearing company which supports the ATM for verifying PINs.  Our last
refinement step hence realises the $\mathit{ATM}$ specification by two
parallel components, represented by the operational specification
$\mathit{ATM}'$ in \cref{fig:op-spec-ATM-ref} and the operational
specification $\mathit{CC}$ of a clearing company
in~\cref{fig:op-spec-CC}.  Both communicate (via shared events) when an
ATM sends a verification request, modelled by the event
$\mathsf{verifyPIN}$, to the clearing company.  The clearing company may
answer with $\mathsf{correctPIN}$ or $\mathsf{wrongPIN}$ and then the
ATM continues following its specification.  For the implementation
construction we use the parallel composition constructor
$\kappa_\otimes$ from $(\Sig(\mathit{ATM}'),\Sig(\mathit{CC}))$ to
$\Sig(\mathit{ATM}') \otimes \Sig(\mathit{CC})$.  The signature of
$\mathit{CC}$ consists of the events shown on the transitions
in~\cref{fig:op-spec-CC}. Moreover, there is one integer-valued
attribute $\mathsf{cnt}$ counting the number of verification tasks
performed.  The signature of $\mathit{ATM}'$ extends
$\Sig(\mathit{ATM})$ by the events $\mathsf{verifyPIN}$,
$\mathsf{correctPIN}$ and $\mathsf{wrongPIN}$.  To fit the signature and
the behaviour of the parallel composition of $\mathit{ATM}'$ and
$\mathit{CC}$ to the specification $\mathit{ATM}$ we must therefore
compose $\kappa_\otimes$ with an event refinement constructor
$\kappa_\alpha$ such that $\alpha(\mathsf{enterPIN}) =
(\mathsf{enterPIN}; \mathsf{verifyPIN}; (\mathsf{correctPIN} +
\mathsf{wrongPIN}))$; for the other events $\alpha$ is the identity and
for the attributes the inclusion.  The idea is therefore that the
refinement looks like $\mathit{ATM}
\refinesto_{\kappa_\otimes;\,\kappa_\alpha} \langle\mathit{ATM}',
\mathit{CC}\rangle$.  To prove this refinement relation we rely on the
syntactic parallel composition $\mathit{ATM}' \parallel \mathit{CC}$
shown in \cref{fig:op-spec-ATM-CC-comp}, and on \cref{thm:par}.  It is
easy to see that $\mathit{ATM} \refinesto_{\kappa_\alpha}
\mathit{ATM}'\parallel\mathit{CC}$. In fact, all transitions for event
$\mathsf{enterPIN}$ in \cref{fig:op-spec-ATM} are split into several
transitions in \cref{fig:op-spec-ATM-CC-comp} according to the event
refinement defined by $\alpha$.  For instance, the loop transition from
$\mathit{PIN}$ to $\mathit{PIN}$ with precondition $\mathsf{trls < 2}$
in \cref{fig:op-spec-ATM} is split into the cycle from $(\mathit{PIN},
\mathit{Idle})$ via $(\mathit{PINEntered}, \mathit{Idle})$ and
$(\mathit{Verifying}, \mathit{Busy})$ back to $(\mathit{PIN},
\mathit{Idle})$ in \cref{fig:op-spec-ATM-CC-comp}.  Thus, we have
$\mathit{ATM} \refinesto_{\kappa_\alpha}
\mathit{ATM}'\parallel\mathit{CC}$ and can apply \cref{thm:par} such
that we get $\mathit{ATM} \refinesto_{\kappa_\otimes;\,\kappa_\alpha}
\langle\mathit{ATM}', \mathit{CC}\rangle$.
\end{example}

%%% Local Variables:
%%% mode: LaTeX
%%% mode: TeX-PDF
%%% mode: TeX-source-correlate
%%% TeX-master: "edts-hdl.tex"
%%% End:

%!TEX root = edts-hdl.tex

\section{Conclusions}\label{sec:conclusions}

We have presented a novel logic, called $\EDHL$"=logic, for the rigorous
formal development of event-based systems incorporating changing data
states. To the best of our knowledge, no other logic supports the full
development process for this kind of systems ranging from abstract
requirements specifications, expressible by the dynamic logic features,
to the concrete specification of implementations, expressible by the
hybrid part of the logic.

The temporal logic of actions (TLA~\cite{lamport:2003}) supports also
stepwise refinement where state transition predicates are considered as
actions.  In contrast to TLA we model also the events which cause data
state transitions.  For writing concrete specifications we have proposed
an operational specification format capturing (at least parts of)
similar formalisms, like Event-B~\cite{abrial:2013}, symbolic transition
systems~\cite{poizat-royer:jucs:2006}, and UML protocol state
machines~\cite{uml-2.5}. A significant difference to Event-B machines is
that we distinguish between control and data states, the former being
encoded as data in Event-B. On the other hand, Event-B supports
parameters of events which could be integrated in our logic as well.  An
institution-based semantics of Event-B has been proposed
in~\cite{farrell-monahan-power:wadt:2016} which coincides with our
semantics of operational specifications for the special case of
deterministic state transition predicates. Similarly, our semantics of
operational specifications coincides with the unfolding of symbolic
transition systems in~\cite{poizat-royer:jucs:2006} if we instantiate
our generic data domain with algebraic specifications of data types (and
consider again only deterministic state transition predicates).  The
syntax of UML protocol state machines is about the same as the one of
operational event/data specifications.  As a consequence, all of the
aforementioned concrete specification formalisms (and several others)
would be appropriate candidates for integration into a development
process based on $\EDHL$"=logic.

There remain several interesting tasks for future research.  First, our
logic is not yet equipped with a proof system for deriving consequences
of specifications.  This would also support the proof of refinement
steps which is currently achieved by purely semantic reasoning.  A proof
system for $\EDHL$"=logic must cover dynamic and hybrid logic parts at
the same time, like the proof system in~\cite{madeira-et-al:tcs:2018},
which, however, does not consider data states, and the recent calculus
of~\cite{bohrer-platzer:lics:2018}, which extends differential dynamic
logic but does not deal with events and reactions to events.  Both proof
systems could be appropriate candidates for incorporating the features
of $\EDHL$"=logic.  Another issue concerns the separation of events into
input and output as in I/O-automata~\cite{lynch:concur:2003}.  Then also
communication compatibility (see~\cite{alfaro-henzinger:sigsoft:2001}
for interface automata without data
and~\cite{bauer-hennicker-wirsing:tcs:2011} for interface theories with
data) would become relevant when applying a parallel composition
constructor.

%%% Local Variables:
%%% mode: LaTeX
%%% mode: TeX-PDF
%%% mode: TeX-source-correlate
%%% TeX-master: "edts-hdl.tex"
%%% End:

\bibliographystyle{splncs04}
\bibliography{bibliography}

\clearpage
\begin{appendix}
%!TEX root = edts-hdl.tex

\section{Proofs}\label{app:proofs}

In order to prove \cref{thm:bisim-invariance} we have to consider the
following result:

\begin{lemma}\label{lem:bisim-invariance}
Let $M_1, M_2$ be $\Sigma$-edts and $B \subseteq \Conf(M_1)
\times \Conf(M_2)$ a bisimulation between $M_1$ and $M_2$, let $\varrho$
be a hybrid-free sentence, and $(\gamma_1, \gamma_2) \in B$.  Then $M_1,
\gamma_1 \EDHLmodels{\Sigma} \varrho$ iff $M_2, \gamma_2
\EDHLmodels{\Sigma} \varrho$.
\end{lemma}
\begin{proof}%[of \cref{lem:bisim-invariance}]
We proceed by induction over the structure of hybrid-free sentences.
For base sentences $\varphi$, by definition $M_1, \gamma_1
\EDHLmodels{\Sigma} \varphi$ iff $\omega(\gamma_1)\datamodels{A(\Sigma)}
\varphi$. Since $(\gamma_1, \gamma_2) \in B$, we have by
\cref{it:bisim-atom} that $\omega(\gamma_2) \datamodels{A(\Sigma)}
\varphi$, i.e., $M_2, \gamma_2 \EDHLmodels{\Sigma} \varphi$.  For
sentences $\dldia{\lambda} \varrho$ we use a well-known result: dynamic
logic constructors are \emph{safe for bisimulation} (see
\cite{VanBenthem1998}), i.e., the \cref{it:bisim-zig} and
\cref{it:bisim-zag} properties of $B$ are preserved from atomic actions
to composed $\Sigma$-ed actions (provable by induction on the structure
of $\Sigma$-ed actions).  Hence, $M_1, \gamma_1 \EDHLmodels{\Sigma}
\dldia{\lambda}\varrho$ iff $M_1, \gamma_1' \EDHLmodels{\Sigma} \varrho$
for some $(\gamma_1, \gamma_1') \in \Rel(M_1)_{\lambda}$. Moreover, since
$(\gamma_1, \gamma_1') \in B$, \cref{it:bisim-zig} ensures the existence
of a $\gamma_2'$ such that $(\gamma_2, \gamma_2') \in \Rel(M_2)_{\lambda}$
and $(\gamma_2, \gamma_2') \in B$.  By induction hypothesis, $M_2,
\gamma_2' \EDHLmodels{\Sigma} \varrho$ and hence, $M_2, \gamma_2
\EDHLmodels{\Sigma} \dldia{\lambda}\varrho$. The converse implication is
analogously proved using the \cref{it:bisim-zag} property.  The proof
for the remaining cases is straightforward.
\end{proof}

\begin{proof}[of \cref{thm:bisim-invariance}]
Since $M_1\bisim M_2$, there is a bisimulation $B \subseteq \Conf(M_1)
\times \Conf(M_2)$ that relates $\Conf_0(M_1)$ and $\Conf_0(M_2)$
according to the \cref{it:bisim-init} property. By supposing $M_1
\EDHLmodels{\Sigma} \varrho$ and $\gamma_2 \in \Conf_0(M_2)$, we have by
\cref{it:bisim-init} that there is a $\gamma_1 \in \Conf_0(M_1)$ such
that $(\gamma_1, \gamma_2)\in B$. By \cref{lem:bisim-invariance}, $M_2,
\gamma_2 \EDHLmodels{\Sigma} \varrho$ and therefore $M_2
\EDHLmodels{\Sigma} \varrho$. The converse direction can be proved
symmetrically.
\end{proof}

\begin{proof}[of \cref{cor:bisim-closed}]
Let $\Spec$ be an axiomatic specification and $M \in \Mod(\Spec)$. By
the definition of $\Mod(\Spec)$, $M \EDHLmodels{Sig(\Spec)}
\Ax(\Spec)$. By \cref{lem:bisim-invariance}, for any $M' \in
\EDHLStr(\Sig(\Spec))$, if $M \bisim M'$ then $M'
\EDHLmodels{Sig(\Spec)} \Ax(\Spec)$, i.e., $M' \in \Mod(\Spec)$.
\end{proof}

\begin{proof}[of \cref{bisinvariance2}]
Let $M_1, \gamma_1 \EDHLmodels{\Sigma} \varrho$ iff $M_2, \gamma_2
\EDHLmodels{\Sigma} \varrho$ hold for all hybrid-free sentences
$\varrho$.  We show the existence of a bisimulation relating $\gamma_1$
and $\gamma_2$.  Consider the relation $B = \{ (\gamma_1,\gamma_2) \mid
M_1,\gamma_1 \EDHLmodels{\Sigma} \varrho \text{ iff } M_2, \gamma_2
\EDHLmodels{\Sigma} \varrho \text{ for any hybrid-free $\Sigma$-ed
  sentence $\varrho$} \}$. Property \cref{it:bisim-atom} holds by
assumption. In order to prove \cref{it:bisim-zig}, assume that there is
a $\gamma'_1\in \Conf(M_1)$ with $(\gamma_1,\gamma_1')\in
\Rel(M_1)_{\lab{e}{\psi}}$, for which, there is no $\gamma_2'\in
\Conf(M_2)$ such that $(\gamma_2,\gamma_2')\in
\Rel(M_2)_{\lab{e}{\psi}}$ and $(\gamma_1',\gamma_2')\in B$.  Then, for
any configuration $\gamma_2' \in \Gamma_2^{\lab{e}{\psi}} = \{
\gamma_2^{\lab{e}{\psi}} \mid (\gamma_2,\allowbreak
\gamma_2^{\lab{e}{\psi}}) \in \Rel(M)_{\lab{e}{\psi}} \}$,
%$\{\gamma_2^i \mid i \in I,\ (\gamma_2,\gamma_2^i)\in \Rel(M_2)_{\lab{e}{\psi}}\}$,
there is a hybrid-free sentence $\varrho_{\gamma_2'}$ such that $M_1,
\gamma'_1 \EDHLmodels{\Sigma} \varrho_{\gamma_2'}$ and $M_2, \gamma_2'
\not\EDHLmodels{\Sigma} \varrho_{\gamma_2'}$.  Since
$\Gamma_2^{\lab{e}{\psi}} \subseteq \{ \gamma_2^e \mid (\gamma_2,
\gamma_2^e) \in \Rel(M_2)_e \}$ is finite by the image-finiteness of
$M_2$, $\varrho = \bigwedge_{\gamma_2' \in \Gamma_2^{\lab{e}{\psi}}}
\varrho_{\gamma_2'}$ is a hybrid-free sentence.  Hence,
$M_1,\gamma_1\EDHLmodels{\Sigma} \dldia{\lab{e}{\psi}} \varrho$ and
$M_2,\gamma_2\not\EDHLmodels{\Sigma} \dldia{\lab{e}{\psi}} \varrho$ what
contradicts the hypothesis that $(\gamma_1,\gamma_2)\in B$. Therefore,
there is no $\gamma_1'$ is these conditions, i.e., \cref{it:bisim-zig}
holds. The \cref{it:bisim-zag} condition is shown in a similar way. It
is proved that $B$ is a bisimulation.
\end{proof}

\begin{proof}[of \cref{prop:op-ed-spec-parallel}]
Let $O_i = (\Sigma_i, C_i, T_i, (c_{i, 0}, \varphi_{i, 0}))$ with
$\Sigma_i = (E_i, A_i)$ for $i \in \{ 1, 2 \}$ such that $A_1 \cap A_2 =
\emptyset$, let $\Sigma = \Sigma_1 \otimes \Sigma_2 = (E_1 \cup E_2, A_1
\cup A_2) = (E, A)$, and let $O = O_1 \parallel O_2 = (\Sigma, C, T, (c_0,
\varphi_0))$.  Let $M_i \in \Mod(O_i)$ for $i \in \{ 1, 2 \}$ and $M =
M_1 \otimes M_2$; we prove that $M \in \Mod(O)$.

\smallskip%
We first show that for all $((c_1, c_2), \varphi, e, \psi, (c_1', c_2'))
\in \Trans(O)$ and $\omega \in \DataSt(A)$ with $\omega \datamodels{A}
\varphi$, there is some $(\gamma, \gamma') \in \Rel(M)_{e_1}$ with
$(\datast(\gamma), \datast(\gamma')) \datamodels{A} \psi$ by induction
over the reachability of $\CtrlSt(O)$: Let $(c_1, c_2) \in \CtrlSt(O)$
with $((c_{1, i}, c_{2, i}), \varphi_{i+1}, e_{i+1},
\psi_{i+1},\allowbreak (c_{1, i+1},\allowbreak c_{2, i+1})) \in
\Trans(O)$ and $\omega_{i+1} \datamodels{A} \varphi_{i+1}$ with
$(\gamma_i, \gamma_{i+1}) \in \Rel(M)_{e_{i+1}}$ such that
$(\omega(\gamma_i), \omega(\gamma_{i+1})) \datamodels{A} \psi_{i+1}$.
Let $((c_1, c_2), \varphi, e, \psi, (c_1', c_2')) \in \Trans(O)$ and
$\omega \in \DataSt(A)$ with $\omega \datamodels{A} \varphi$ be given.

\noindent%
$e \in E_1 \setminus E_2$: Then $(c_1, \varphi, e, \psi_1, c_1') \in
\Trans(O_1)$ with $\psi = \psi_1 \land \id{A_2}$, and
$\restrict{\omega}{A_1} \datamodels{A_1} \varphi$.  As $M_1 \in
\Mod(O_1)$, there is a $(\gamma_1, \gamma_1') \in \Rel(M_1)_e$ such that
$(\datast(\gamma_1), \datast(\gamma_1')) \datamodels{A_1} \psi_1$.  Let
$\gamma = ((c_1, c_2), \datast(\gamma_1) + \restrict{\omega}{A_2})$ and
$\gamma' = ((c_1', c_2), \datast(\gamma_1') + \restrict{\omega}{A_2})$;
then $(\datast(\gamma), \datast(\gamma') \datamodels{A} \psi$.  By
induction hypothesis, $\gamma \in \Conf(M)$ and hence $(\gamma, \gamma')
\in \Rel(M)_e$.

\noindent%
$e \in E_2 \setminus E_1$: Symmetric to $e \in E_1 \setminus E_2$.

\noindent%
$e \in E_1 \cap E_2$: Then $(c_1, \varphi_1, e, \psi_1, c_1') \in
\Trans(O_1)$, $(c_2, \varphi_2, e, \psi_2, c_2') \in \Trans(O_2)$ with
$\varphi = \varphi_1 \land \varphi_2$, $\psi = \psi_1 \land \psi_2$,
$\restrict{\omega}{A_1} \datamodels{A_1} \varphi_1$,
$\restrict{\omega}{A_2} \datamodels{A_2} \varphi_2$.  Since $M_i \in
\Mod(O_i)$, there are $(\gamma_i, \gamma_i') \in \Rel(M_i)_e$ such that
$(\datast(\gamma_i), \datast(\gamma_i')) \datamodels{A_i} \psi_i$ for $i
\in \{ 1, 2 \}$.  Let $\gamma = ((c_1, c_2), \datast(\gamma_1) +
\datast(\gamma_2))$ and $\gamma' = ((c_1', c_2'), \datast(\gamma_1') +
\datast(\gamma_2'))$; then $(\datast(\gamma), \datast(\gamma')
\datamodels{A} \psi$.  By induction hypothesis, $\gamma \in \Conf(M)$
and hence $(\gamma, \gamma') \in \Rel(M)_e$.

\smallskip%
We now show that for all $e \in E$ and $(((c_1, c_2), \omega), ((c_1',
c_2'), \omega')) \in \Rel(M)_e$ there is some $((c_1, c_2), \varphi, e,
\psi, (c_1', c_2')) \in \Trans(O)$ with $\omega \datamodels{A} \varphi$ and
$(\omega, \omega') \datamodels{A} \psi$ by induction over the
reachability of $\Conf(M)$: Let $((c_1, c_2), \omega) \in \Conf(M)$ with
$(((c_{1, i}, c_{2, i}), \omega_i), ((c_{1, i+1},\allowbreak c_{2,
  i+1}),\allowbreak \omega_{i+1})) \in \Rel(M)_{e_{i+1}}$ such that there
are $((c_{1, i}, c_{2, i}), \varphi_{i+1}, e_{i+1},
\psi_{i+1},\allowbreak (c_{1, i+1},\allowbreak c_{2, i+1})) \in T(O)$
with $\omega_i \datamodels{A} \varphi_i$ and $(\omega_i, \omega_{i+1})
\datamodels{A} \psi_{i+1}$ for all $0 \leq i < n$ and $((c_{1,
  n},\allowbreak c_{2, n}),\allowbreak \omega_n) = ((c_1, c_2),
\omega)$.  Let $(((c_1, c_2), \omega), ((c_1', c_2'), \omega')) \in
\Rel(M)_e$.

\noindent%
$e \in E_1 \setminus E_2$: Then $((c_1, \restrict{\omega}{A_1}), (c_1',
\restrict{\omega}{A_1})) \in \Rel(M_1)_e$ and $c_2' = c_2$.  Since $M_1 \in
\Mod(O_1)$, there is some $(c_1, \varphi_1, e, \psi_1, c_1') \in \Trans(O_1)$
with $\restrict{\omega}{A_1} \datamodels{A_1} \varphi_1$ and
$(\restrict{\omega}{A_1},\allowbreak \restrict{\omega'}{A_1})
\datamodels{A_1} \psi_1$.  By induction hypothesis, $(c_1, c_2) \in
\CtrlSt(O)$ and hence $((c_1, c_2), \varphi_1, e,\allowbreak \psi_1 \land
\id{A_2}, (c_1', c_2)) \in \Trans(O)$, where $\omega \datamodels{A}
\varphi_1$ and $(\omega, \omega') \datamodels{A} \psi_1 \land \id{A_2}$.

\noindent%
$e \in E_2 \setminus E_1$: Symmetric to $e \in E_1 \setminus E_2$.

\noindent%
$e \in E_1 \cap E_2$: Then $((c_1, \restrict{\omega}{A_1}), (c_1',
\restrict{\omega}{A_1})) \in \Rel(M_1)_e$ and $((c_2,
\restrict{\omega}{A_2}), (c_2',\allowbreak \restrict{\omega}{A_2})) \in
\Rel(M_2)_e$.  Since $M_i \in \Mod(O_i)$, there are some $(c_i,
\varphi_i, e, \psi_i, c_i') \in \Trans(O_i)$ such that
$\restrict{\omega}{A_i} \datamodels{A_i} \varphi_i$ and
$(\restrict{\omega}{A_i}, \restrict{\omega'}{A_i}) \datamodels{A_i}
\psi_i$ for $i \in \{ 1, 2 \}$.  By induction hypothesis, $(c_1, c_2)
\in \CtrlSt(O)$ and hence $((c_1, c_2), \varphi_1 \land \varphi_2, e,
\psi_1 \land \psi_2, (c_1', c_2')) \in \Trans(O)$, where $\omega
\datamodels{A} \varphi_1 \land \varphi_2$ and $(\omega, \omega')
\datamodels{A} \psi_1 \land \psi_2$.
\end{proof}

\begin{proof}[of \cref{thm:par}]
Let $\Spec \refinesto_{\kappa} O_1 \parallel O_2$ hold, i.e., $\kappa(M)
\in \Mod(\Spec)$ for all $M \in \Mod(O_1 \parallel O_2)$.  Let $M_1 \in
\Mod(O_1)$ and $M_2 \in \Mod(O_2)$.  Then $\kappa_{\otimes}(M_1, M_2) =
M_1 \otimes M_2 \in \Mod(O_1 \parallel O_2)$ by
\cref{prop:op-ed-spec-parallel}, that is, $\kappa(\kappa_{\otimes}(M_1,
M_2)) \in \Mod(\Spec)$, and thus $\Spec \refinesto_{\kappa_{\otimes};
  \kappa} O_1 \parallel O_2$.
\end{proof}

%%% Local Variables:
%%% mode: LaTeX
%%% mode: TeX-PDF
%%% mode: TeX-source-correlate
%%% TeX-master: "edts-hdl.tex"
%%% End:

\end{appendix}

\end{document}